\documentclass[11pt]{article}

\usepackage[a4paper]{geometry}
\usepackage{amsthm}
\usepackage{xcolor}
\usepackage{amsmath}
\usepackage{amssymb}
\usepackage{tikz}
\usepackage{microtype}
\usepackage[english]{babel}
\usepackage{enumerate}
\usepackage{todonotes}
\usepackage{graphicx}
\usepackage{hyperref}
\usetikzlibrary{matrix}
\usepackage{graphicx}
\usepackage{float}
\usepackage{listings}
\usepackage{amsfonts}
\usepackage[all]{xy}
\usepackage{mathtools}
\usepackage{caption}
\usepackage{subcaption}
\usepackage{comment}
\usepackage{bm}
\usepackage{csquotes}
\usepackage{multirow}
\usepackage{makecell}
\usepackage{booktabs}
\usepackage{array}
\usepackage{tabularray}

\usepackage{arydshln}
\usepackage[font=footnotesize,labelfont=bf]{caption}
\usepackage{authblk}
\usepackage{lineno}

\usepackage{todonotes}

\usepackage{empheq}
\usepackage[most]{tcolorbox}

\newtcbox{\mymath}[1][]{
    nobeforeafter, math upper, tcbox raise base,
    enhanced, colframe=black!30!black,
    colback=green!30, boxrule=1pt,
    #1}

\usepackage[backend=biber,citestyle=numeric-comp,sorting=none,doi=true,isbn=false,url=false,eprint=false]{biblatex}
\renewbibmacro{in:}{}
\numberwithin{equation}{section}

\newtheorem{theorem}{Theorem}[section]
\newtheorem{corollary}[theorem]{Corollary}
\newtheorem{lemma}[theorem]{Lemma}

\newtheorem{definition}{Definition}[section]
\newtheorem{remark}{Remark}[section]

\newtheorem{assumption}{Assumption}[section]
\theoremstyle{definition} 

\title{Pattern formation in a Swift-Hohenberg equation with\\spatially periodic coefficients}
\author{Jolien Kamphuis, Martina Chirilus-Bruckner}
\affil{\textit{Mathematical Institute, Leiden University, The Netherlands}}

\date{}

\begin{document}

\maketitle

\begin{abstract}
We study the Swift–Hohenberg equation -- a paradigm model for pattern formation -- with ``large" spatially periodic coefficients and find a Turing bifurcation that generates patterns whose leading‑order form is a Bloch wave modulated by solutions of a Ginzburg–Landau type equation. Since the interplay between forcing wavenumber and intrinsic wavenumber crucially shapes the spectrum and emerging patterns, we distinguish between resonant and non‑resonant regimes. Extending earlier work that assumed asymptotically small coefficients, we tackle the more involved onset analysis produced by $\mathcal{O}(1)$ forcing and work directly in Bloch space, where the richer structure of the bifurcating solutions becomes apparent. This abstract framework is readily transferable to more complex systems, such as reaction–diffusion equations arising as dryland vegetation models, where topography induces spatial heterogeneity.
\end{abstract}

\tableofcontents

\newpage

\section{Introduction}
Consider the Swift-Hohenberg equation
\begin{equation}\label{EQ:SH_periodic_coefficients}
    \partial_t u = \left[- (\partial_x^2+k_0^2)^2 + p(x) \right] u - \rho(x)u^3 \, ,
\end{equation}
for the unknown $u=u(x,t) \in\mathbb{R}$ depending on $x\in\mathbb{R}, t\geq 0,$ and given $k_0 > 0 $ and periodic coefficients
\begin{align}\label{EQ:periodic_ceofficients}
p\left(x+\frac{2\pi}{k_f}\right) = p(x) \, , \quad \rho\left(x+\frac{2\pi}{k_f}\right) = \rho(x) \, , \quad x \in \mathbb{R} \,,
\end{align}
with period $\frac{2\pi}{k_f}$ for some frequency $k_f > 0$.

\medskip

The classically studied Swift-Hohenberg equation (see \cite{Swift_Hohenberg}) can be recovered by setting $p(x) = \varepsilon^2 r$ for some $0 < \varepsilon \ll 1, r \in \mathbb{R}, r = \mathcal{O}(1)$, and $\rho(x) = 1$, $x \in \mathbb{R}$. For this choice the constant background solution $u_* = 0$ undergoes a Turing bifurcation at $r = 0$ giving rise to stable spatially periodic solutions around the wavenumber $k_0 > 0$ for $r > 0$. These are often referred to as {\it patterns}, since they are also directly linked to patterns in the two-dimensional case \cite{Hoyle2006Pattern}. The dynamics of the amplitude of Turing patterns is known to be governed by the Ginzburg-Landau (GL) equation (see, e.g. \cite{diprima_eckhaus_segel_1971})
\begin{equation}
    \partial_T A = r A + 4k_0^2\partial_X^2A-3|A|^2A,
\end{equation}
via the relation to the Swift-Hohenberg given by
\begin{equation}\label{eq:ansatz_cc}
    u(x,t) = u_* + \varepsilon A(X,T) e^{ik_0x} + \rm{c.c.} + \rm{h.o.t.},
\end{equation}
where $X:=\varepsilon x$, $T:=\varepsilon^2 t$ and $\rm{h.o.t.}$ are higher order terms w.r.t. $\varepsilon$. 

\medskip

{\bf Main objective.} The present work studies pattern formation for the case of the Swift-Hohenberg equation with multiplicative spatially periodic forcing as given in \eqref{EQ:SH_periodic_coefficients} by rigorous derivation and justification of the Ginzburg-Landau equation.

\medskip

The evolution of patterns in systems with a slightly varying geometry has been previously studied by deriving the form of a Ginzburg-Landau equation for a class of reaction-diffusion equations with quadratic nonlinearity and small varying terrain, written as $\nu p(k_fx)$, with $\nu=\mathcal{O}(\varepsilon^2)$,  $k_f=\varepsilon \sigma_0$, and $\sigma_0=\mathcal{O}(1)$ in \cite{Eckhaus_Kuske}, and with $\nu=\mathcal{O}(\varepsilon)$ and $k_f=\mathcal{O}(1)$ in \cite{Doelman_Schielen}. These works were motivated by generalizing sedimentation in straight rivers, treated in \cite{schielen_doelman_deswart}, to slightly meandering rivers. Ginzburg-Landau equations have also been derived for the Swift-Hohenberg model with small spatial forcing $p(x) = \varepsilon^2 r + \gamma \cos(k_f x)$ in \cite{Ehud_2008, Ehud_2012, Ehud_spatialperforcing, Ehud_2014, Ehud_2015_1, uecker2001}, where the forcing strength is taken of order $\gamma = \mathcal{O}(\varepsilon^k)$ for $k=1,2$. Furthermore, \cite{Ehud_2015_2, Ehud_2014_CDIMA} considers a spatial forcing in the Lengyel-Epstein model. We will elaborate more on related work in Section~\ref{sec:related_work}.

\medskip

{\bf Main novelty.} Unlike previous studies, we allow forcing of order one, $p(x)=\mathcal{O}(1)$, whose amplitude and wavelength are neither small nor suited to classical averaging techniques. This regime introduces subtler bifurcation analysis due to the necessity of Bloch analysis at bifurcation onset, but it better reflects many applications. In particular, in the context of vegetation patterns, an application that partially motivated our work, $p$ and $\rho$ capture topographic effects. The abstract framework used here is readily transferable to more complicated settings relevant to applications and, in particular, lay groundwork for studying more complex PDEs such as ecologically motivated reaction–diffusion systems (cf. e.g. \cite{Bastiaansen_CH4}) and delivers the $p(x)=\mathcal{O}(1)$ result that was not treated in \cite{Ehud_spatialperforcing}. Moreover, we give a validity proof and error estimates. We focus on the interaction between the system’s intrinsic wavenumber $k_0$ and the forcing wavenumber $k_f$ of the periodic coefficients $p$ and $\rho$. The interplay between forcing strength, wavenumber $k_f$ and intrinsic wavenumber $k_0$ will give rise to {\it resonant} and {\it non‑resonant} parameter regimes.

\subsection{Linear analysis}
In the classic case $p(x) = \varepsilon^2 r$ for some $0 < \varepsilon \ll 1, r \in \mathbb{R}, r = \mathcal{O}(1)$ and $\rho(x) = 1, x \in \mathbb{R}$, a Turing instability is easily anticipated from the fact that all spatially bounded solutions of the linearized Swift-Hohenberg equation can be built from $v(x,t)  = e^{ikx}e^{\lambda(k)t}$, $\lambda(k) = -(k^2 - k_0^2)^2 + \varepsilon^2 r$, where we depicted the so-called dispersion relation $(k, \lambda(k))$ where $\lambda(k) = \lambda(k;k_0,r)$ in Figure~\ref{fig:SH_cc_spec}. This shows that $u_*=0$ is linearly stable for $r<0$ and linearly unstable for $r>0$. Upon adding nonlinear terms, this instability creates stable periodic solutions with wavenumber close to $k_0$.  Note that, upon setting $p(x) = \varepsilon^2 r + p_*(x)$ with $p_*(x) = \mathcal{O}(\varepsilon^m), m > 0,$ and $k_f$ non-resonant, the linearized problem at onset is to leading order in $\varepsilon$ still the same as above, which makes the derivation of the GL equation similar to the constant-coefficient case.

\begin{figure}[ht]
\centering
\begin{subfigure}[]{0.32\textwidth}
    \includegraphics[width=\textwidth]{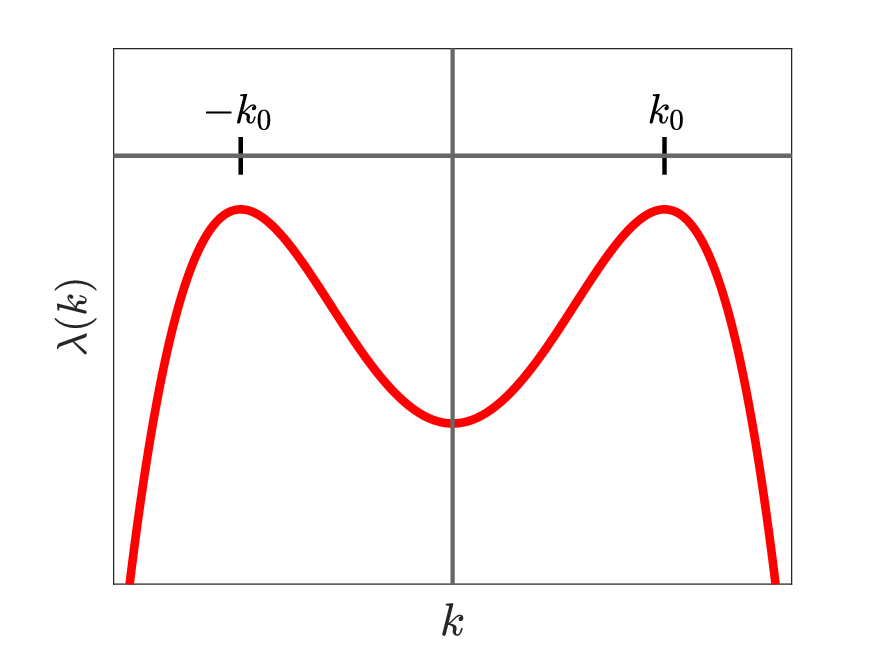}
    \caption{Spectrum before onset}
  \end{subfigure}
    \begin{subfigure}[]{0.32\textwidth}
    \includegraphics[width=\textwidth]{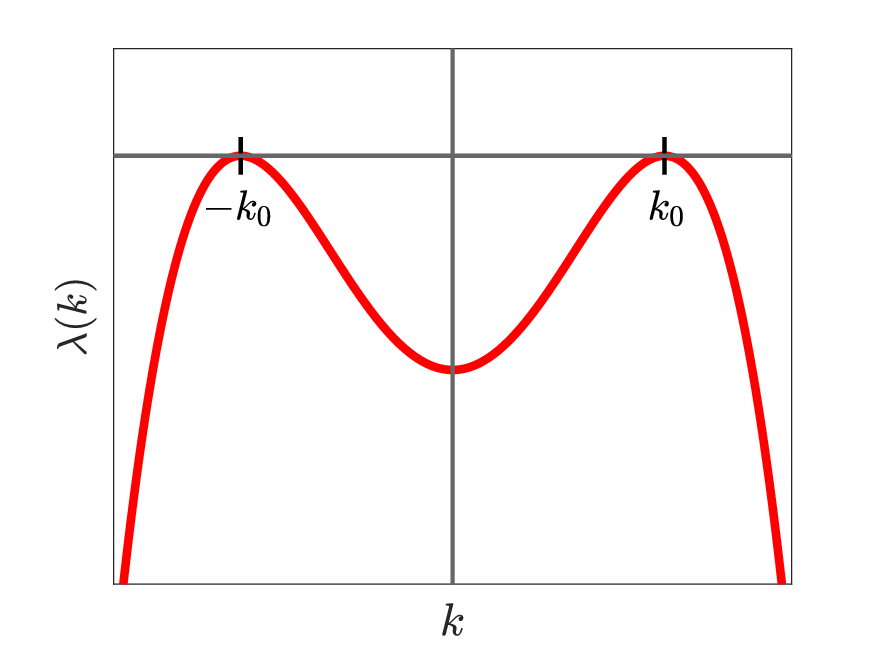}
    \caption{Spectrum at onset}
    \end{subfigure}
    \begin{subfigure}[]{0.32\textwidth}
    \includegraphics[width=\textwidth]{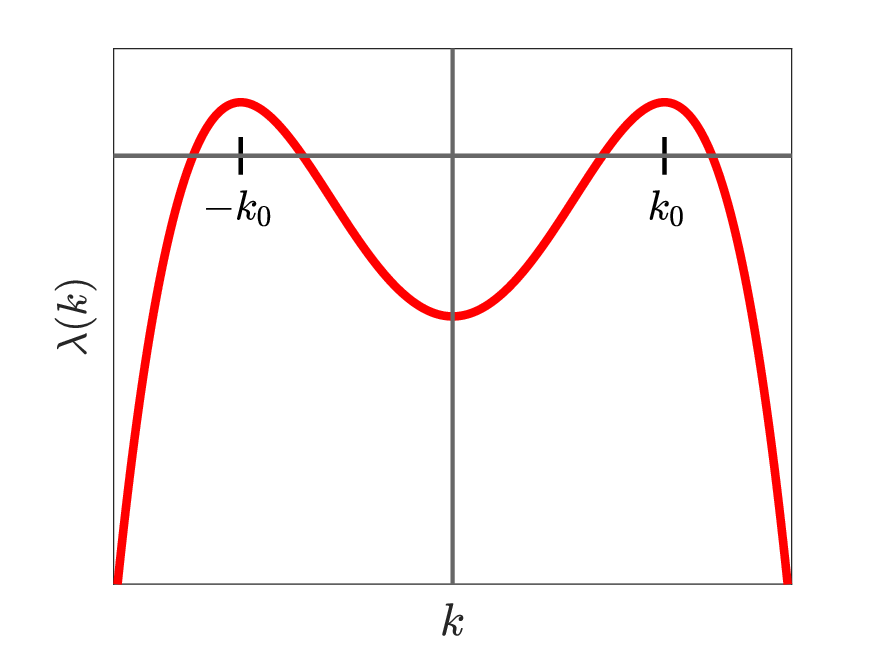}
    \caption{Spectrum after onset}
    \end{subfigure}
  \caption{Spectrum of the linear operator of the constant coefficient Swift-Hohenberg system around $u_*=0$. Before onset, $r=-1$ (left), at onset, $r=0$ (middle) and after bifurcation, $r=1$ (right).}
  \label{fig:SH_cc_spec}
\end{figure}

In stark contrast to this, the case of $p(x) = \varepsilon^2 r + p_*(x)$ with $p_*(x) = \mathcal{O}(1)$ is more complicated, since spatially bounded solutions of the linearized PDE involving
\begin{align}\label{EQ:linear_operator}
    \mathcal{L}_p:= -\left(\partial_x^2 - k_0^2\right)^2 + p(x)
\end{align}
are built from Floquet-Bloch solutions (cf. \cite{Simon_Reed})
\begin{align}\label{EQ:Floquet-Bloch_solutions}
 v(x,t)  = e^{ilx}\psi_n(\ell, x)e^{\lambda_n(\ell)t}\,, \quad n \in \mathbb{N} \, ,
\end{align}
with some function $\psi_n\left(\ell, x + \frac{2\pi}{k_f}\right) = \psi_n(\ell, x)$ that has the same period as $p, \rho$ and so-called band structure
\begin{align}\label{EQ:bands_periodic_coefficient}
(\ell, \lambda_n(\ell)) \, , \quad n \in \mathbb{N}, -k_f/2 \leq \ell < k_f/2\, ,
\end{align}
with $\lambda_n(\ell) = \lambda_n(\ell;k_0, p)$ which, in general, cannot be computed explicitly (see Section~\ref{sec:piecewise} and \cite{MCB_Breather_2011} for an exception where it can be specified in terms of elementary functions) making numerical explorations indispensable. While the occurrence of Floquet-Bloch solutions is fairly generic and is backed by abstract theory, them having special features as need for a Turing instability is less explored analytically. We refrain from giving a full theoretical characterization of classes of forcing functions $p$ (leaving it to future work) and rather work by assumptions, which also have the benefit of preparing further generalizations of the result. Hence, we will now collect a number of conditions on $p$ that will give rise to a band structure that causes a Turing instability similar to the classic constant-coefficient SH (see Figure~\ref{fig:SH_spec_cartoon} for a cartoon).

\begin{figure}[ht]
\centering
    \begin{subfigure}[]{0.32\textwidth}
    \includegraphics[width=\textwidth]{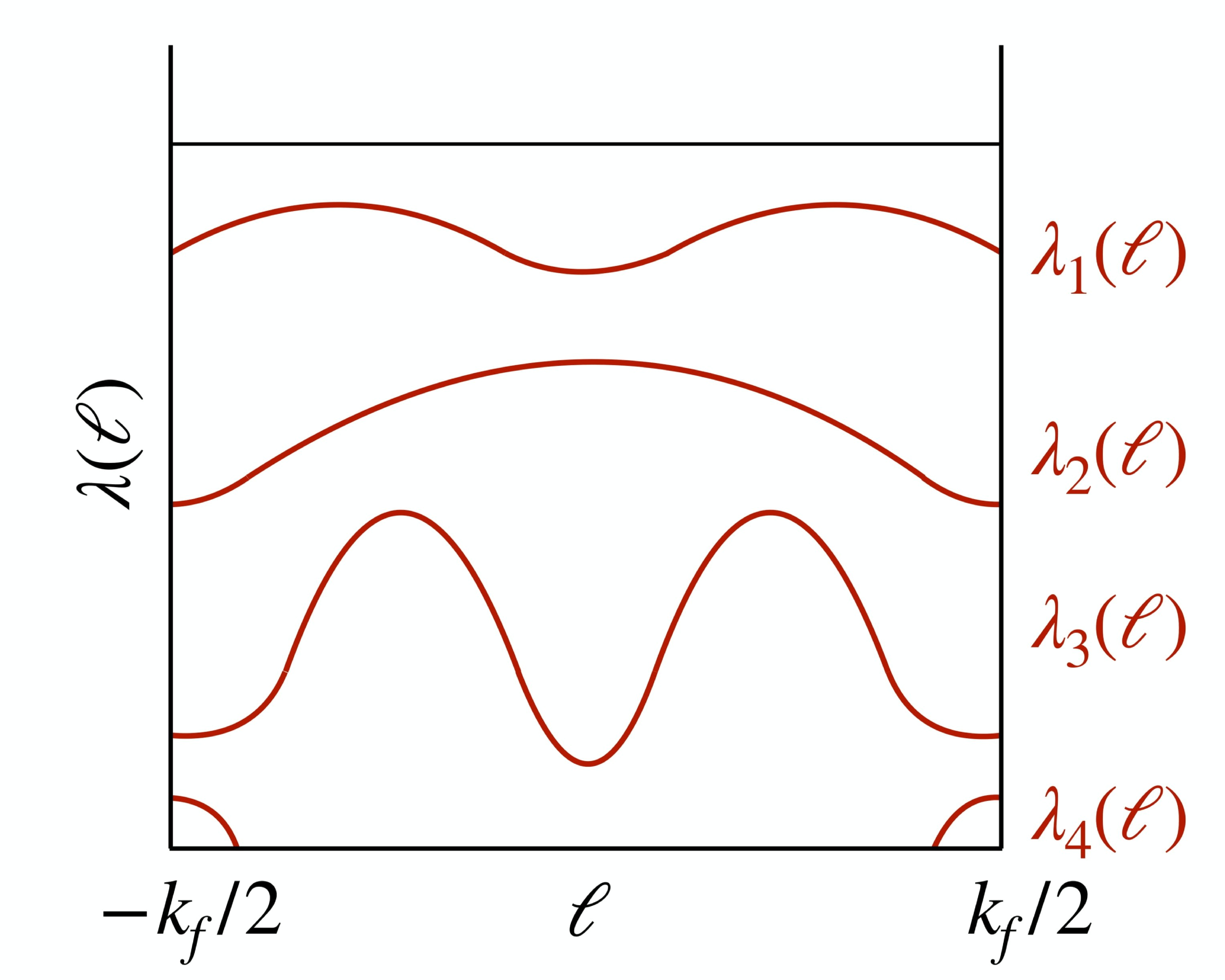}
    \caption{Spectrum before onset}
  \end{subfigure}
    \begin{subfigure}[]{0.32\textwidth}
    \includegraphics[width=\textwidth]{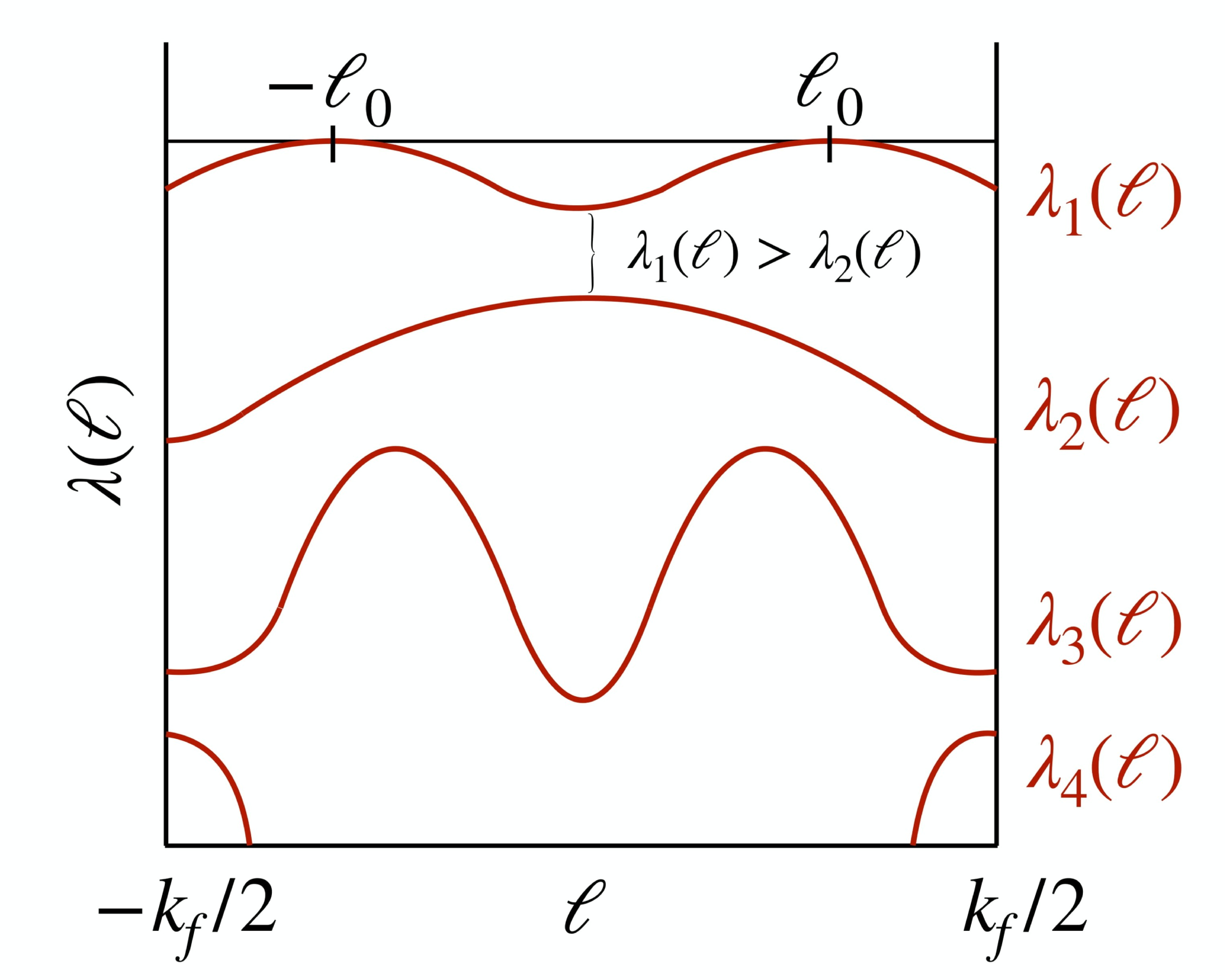}
    \caption{Spectrum at onset}
  \end{subfigure}
    \begin{subfigure}[]{0.32\textwidth}
    \includegraphics[width=\textwidth]{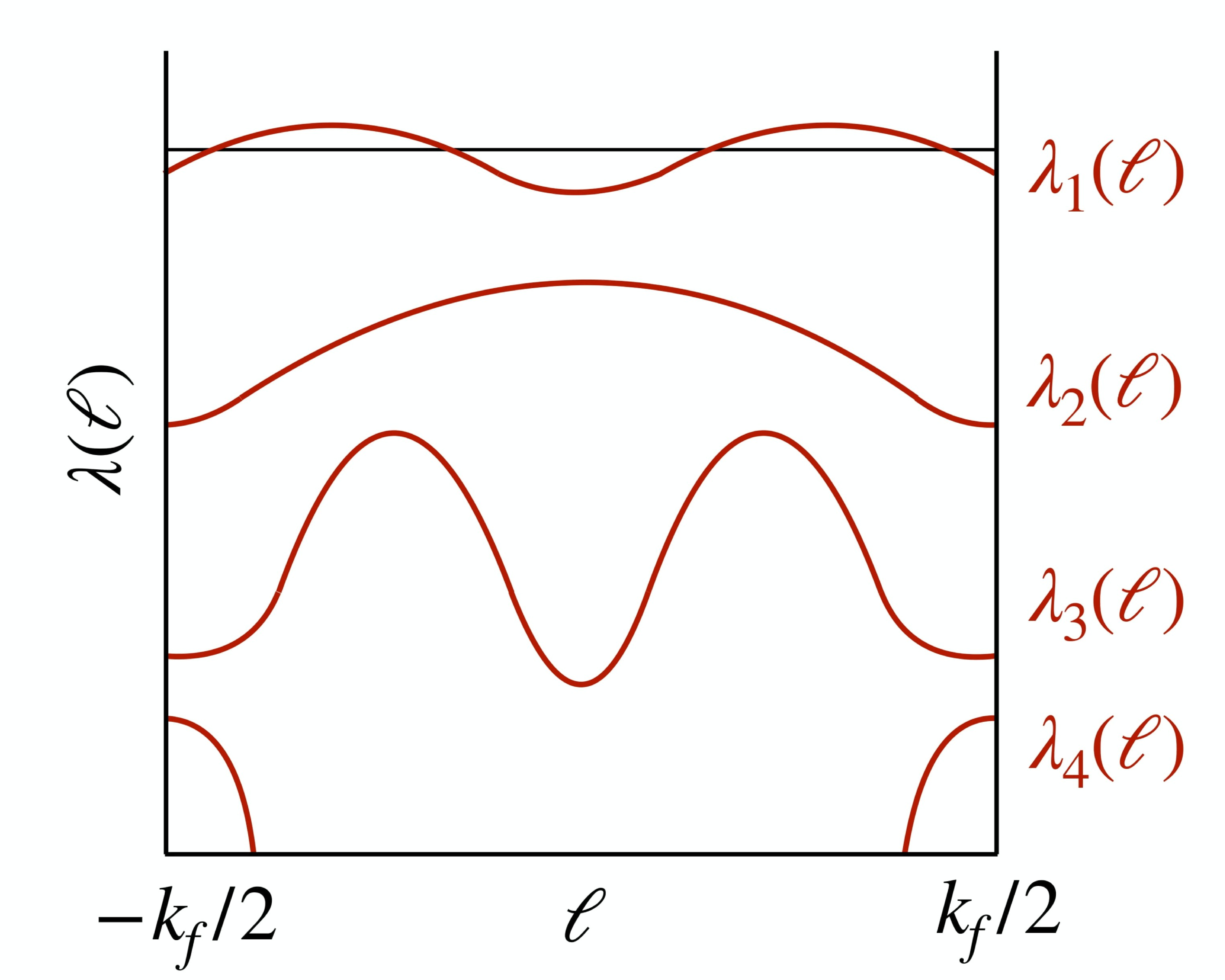}
    \caption{Spectrum after onset}
  \end{subfigure}
\caption{Cartoon of a band structure $(\ell, \lambda_n(\ell))$, satisfying the Turing instability Assumption~\ref{ASSUMP:Turing_instability}: $r < 0, r = 0, r >0$.}
\label{fig:SH_spec_cartoon}
\end{figure}

\subsection{Main result}
Before stating the main result we compile a list of assumptions that it relies on.

\begin{assumption}[Turing instability]\label{ASSUMP:Turing_instability}
 Consider the linear operator $\mathcal{L}_p$ in \eqref{EQ:linear_operator}. We assume that
 \begin{align}
     p(x) = p_*(x) + \varepsilon^2 r \, , \quad r \in \mathbb{R} \, ,
 \end{align}
 for $0 < \varepsilon \ll 1$ and some periodic function $p_*(x) = \mathcal{O}(1)$
 is chosen such that the band structure \eqref{EQ:bands_periodic_coefficient} is given by $(\ell, \lambda_n(\ell))$ with $\lambda_n(\ell) = \lambda_n(\ell;k_0, p_*, \varepsilon, r)$ and the following holds true for the Floquet-Bloch solutions from \eqref{EQ:Floquet-Bloch_solutions}.
 \begin{itemize}
     \item[(a)] {\bf (ONB)} The sequence of functions $\{\psi_n(\ell,\cdot)\}_{n\in\mathbb{N}}$ forms an orthonormal basis of $L^2\left(\left[0, \frac{2\pi}{k_f}\right]\right)$ for every $ -k_f/2 \leq \ell < k_f/2$.
     \item[(b)] {\bf (Ordering of band structure)} We have $\lambda_{n}(\ell) \geq \lambda_{n+1}(\ell), n \in \mathbb{N}, \lambda_{n}(\ell) \rightarrow -\infty, n \rightarrow \infty$ for every $ -k_f/2 \leq \ell < k_f/2 $.
     \item[(c)] {\bf (Regularity properties of spectral data)} The $\psi_n, \lambda_n$ are at least three times differentiable w.r.t. $\ell$, with $\partial_\ell^k \psi_n(\ell,x)\in L^2([0, 2\pi/k_f])$ for $0\leq k\leq 3$ and satisfy $\lambda_n(-\ell) = \lambda_n(\ell)$, $\psi_n(-\ell, x) = \overline{\psi_n(\ell,x)}$ for every $ -k_f/2 \leq \ell < k_f/2 $.
     \item[(d)] {\bf (Bifurcation w.r.t. $\bm{r}$)} The function $p_*$ is such that $ \mathrm{sign}(r) = \mathrm{sign}(\max_{-k_f/2 \leq \ell < k_f/2}(\lambda_{1}(\ell)))$ and, in particular, for $r = 0$ we have that there exists a unique $\ell_0$ with $0 = \lambda_1(\ell_0) = \max_{0 \leq \ell < k_f/2}\lambda_{1}(\ell)$.
     \item[(e)] {\bf (Parabolic band structure at onset)} At bifurcation onset, $r = 0$, we have that
\begin{equation}
    \lambda_1(\ell_0) = 0, \quad \partial_\ell \lambda_1(\ell_0)=0, \quad \partial_\ell^2 \lambda_1(\ell_0)<0.
\end{equation}
    for $l_0 \in [0, k_f/2)$ and $\lambda_1(\ell)$ is of parabola type around $\ell_0$, i.e., $\lambda_1(\ell) \leq -C(\ell-\ell_0)^2$ for some $C>0$ and for $\ell\in B_\varepsilon (\ell_0)$, with $B_\varepsilon(\ell_0)$ an $\varepsilon$-neighborhood of $\ell_0$.
    \item[(f)] {\bf (Spectral gap)} At bifurcation onset, $r = 0$, there is a spectral gap between the first two spectral curves, $\min_\ell \lambda_1(\ell)>\max_\ell \lambda_2(\ell)$.
    \item[(g)] {\bf (Non-resonance conditions)} We have that $\lambda_1(3\ell_0)\neq 0$, $\lambda_1(3\ell_0)= \mathcal{O}(1)$ and $\lambda_2(\ell_0)=\mathcal{O}(1), \lambda_2(3\ell_0)=\mathcal{O}(1)$ with respect to $\varepsilon$.
      \end{itemize}
\end{assumption}

We now have all ingredients to state the main theorem.

\begin{theorem}\label{thm:SH_approx}\emph{\textbf{(Approximation Theorem - Ginzburg-Landau for the non-resonant case)}}
    Consider the Swift-Hohenberg equation \eqref{EQ:SH_periodic_coefficients} with periodic functions $p, \rho$ as in \eqref{EQ:periodic_ceofficients} such that Assumption~\ref{ASSUMP:Turing_instability} holds and $\rho$ is continuous. Let $T_0\in\mathbb{R}$, $T_0>0$ and $T_0=\mathcal{O}(1)$ and let $Y$ be a suitably chosen function space. Then there exists an $\varepsilon_0>0$ and a $C>0$ such that for all $\varepsilon\in(0,\varepsilon_0)$ there are solutions $u$ of the Swift-Hohenberg equation \eqref{EQ:SH_periodic_coefficients} with
    \begin{equation}\label{eq:thm_error_est}
        \sup_{t\in [0,~T_0/\varepsilon^2]} \|u(\cdot,t)-u_{\rm Ans}(\cdot, t)\|_{Y} \leq C\varepsilon^2,
    \end{equation}
    where $u_{\rm Ans} = u_{\rm GL} + \mathcal{O}(\varepsilon^2)$ is to leading order given by
    \begin{equation}\label{eq:u_GL}
        u_{\rm GL}(x,t) = \varepsilon A_1(X,T)\psi_1(\ell_0,x) e^{i\ell_0x}+\varepsilon \overline{A_1(X,T)} \overline{\psi_1(\ell_0,x)} e^{-i\ell_0x},
    \end{equation}
    with $X:=\varepsilon x, ~T:=\varepsilon^2 t$, and where $A_1$ solves the Ginzburg-Landau equation
    \begin{equation}\label{eq:GL_forced_SH}
        \begin{split}
    \partial_T A_1 = & ~r A_1 + \frac{1}{2} \partial_\ell^2 \lambda_1 (\ell_0) \partial_X^2 A_1  - 3 \left(\int_0^{2\pi/k_f} |\psi_1(\ell_0,x)|^4 \rho(x) ~dx \right)~ |A_1|^2 A_1 \, ,
    \end{split}
    \end{equation}
    for non-resonant wavenumbers $k_f$. 
\end{theorem}

When $\ell_0=3\ell_0 \pmod{k_f}$, Assumption~\ref{ASSUMP:Turing_instability}(g) can not be satisfied. In these cases, either $\ell_0=0$ or $\ell_0=\frac{k_f}{2}$, leading to resonances.

\begin{definition}\label{def:resonances}\emph{\textbf{(Resonances)}}
If the spectrum admits a spectral gap at onset between $\lambda_1(\ell)$ and $\lambda_2(\ell)$, the combination of parameters such that $\ell_0 = \frac{k_f}{2}$ is called \emph{resonant}. Combinations of parameters that yield to $\ell_0\neq 3\ell_0\pmod{k_f}$ are referred to as \emph{non-resonant}. 
\end{definition}

The proof of Theorem \ref{thm:SH_approx} yields the following corollaries for values of $k_f$ such that $\ell_0=3\ell_0\pmod{k_f}$. Here, we need an adaptation of Assumption~\ref{ASSUMP:Turing_instability}(g):

\newcommand{\mainTheoremNumber}{\arabic{section}.\arabic{assumption}}
\setcounter{assumption}{\value{assumption}-1}
\begin{assumption}
    \begin{itemize}
        \item [(g')] We have $\lambda_2(\ell_0)=\mathcal{O}(1)$ with respect to $\varepsilon$.
    \end{itemize}
\end{assumption}

\begin{corollary}[\textbf{Allen-Cahn approximation for the $\ell_0=0$ case}]\label{cor:SH_approx_l00}
    Consider parameter regimes such that $\ell_0=0$ and the same assumptions as Theorem \ref{thm:SH_approx} hold, with Assumption~\ref{ASSUMP:Turing_instability}(g) replaced by Assumption~\ref{ASSUMP:Turing_instability}(g'). Let $T_0\in\mathbb{R}$, $T_0>0$ and $T_0=\mathcal{O}(1)$ and let $Y$ be a suitably chosen function space. Then there exists an $\varepsilon_0>0$ and a $C>0$ such that for all $\varepsilon\in(0,\varepsilon_0)$ there are solutions $u$ of the Swift-Hohenberg equation \eqref{EQ:SH_periodic_coefficients} such that \eqref{eq:thm_error_est} holds, where $u_{\rm Ans} = u_{\rm AC} + \mathcal{O}(\varepsilon^2)$ is to leading order given by
    \begin{equation}\label{eq:u_AC}
        u_{\rm AC}(x,t) = \varepsilon A_1(X,T)\psi_1(0,x)\,,
    \end{equation}
    with $X:=\varepsilon x, ~T:=\varepsilon^2 t$, $A_1$, $\psi_1(0,\cdot) \in \mathbb{R}$ and where $A_1$ solves the Allen-Cahn equation
    \begin{equation}\label{eq:AC_forced_SH}
         \begin{split}
     \partial_T {A}_{1} = & ~ rA_1 +  \frac12 \partial_{\ell}^2\lambda_1(0) \partial_X^2A_1 -\left(\int_{0}^{\frac{2 \pi}{k_f}}\psi_1(0,x)^4\rho(x)~ dx\right) A_1^3\,.
     \end{split}
    \end{equation}
\end{corollary}

\begin{corollary}[\textbf{Non-phase invariant GL approximation for a resonant case}]\label{cor:SH_approx_res}
    Consider $k_f$ resonant and let the same conditions as Theorem \ref{thm:SH_approx} hold true, with Assumption~\ref{ASSUMP:Turing_instability}(g) replaced by Assumption~\ref{ASSUMP:Turing_instability}(g'). Let $T_0\in\mathbb{R}$, $T_0>0$ and $T_0=\mathcal{O}(1)$ and let $Y$ be a suitably chosen function space.
    Then there exists an $\varepsilon_0>0$ and a $C>0$ such that for all $\varepsilon\in(0,\varepsilon_0)$ there are solutions $u$ of the Swift-Hohenberg equation \eqref{EQ:SH_periodic_coefficients} such that \eqref{eq:thm_error_est} and \eqref{eq:u_GL} hold, where $A_1$ solves the non-phase invariant Ginzburg-Landau equation
    \begin{equation}\label{eq:GL_forced_SH_res}
         \begin{split}
     \partial_T {A}_{1} = & ~ rA_1 +  \frac12 \partial_{\ell}^2\lambda_1(\ell_0) \partial_X^2A_1 - 3 \left(\int_{0}^{\frac{2 \pi}{k_f}} |\psi_1(\ell_0,x)|^4\rho(x)~dx\right) |A_1|^2 A_1  \\ 
     &\quad -\left(\int_{0}^{\frac{2 \pi}{k_f}}\psi_1(\ell_0,x)^3 \overline{\psi_1( \ell_0,x)}\rho(x)~ dx\right) A_1^3\,.
     \end{split}
    \end{equation}
\end{corollary}

Solutions $u$ of the Swift-Hohenberg \eqref{EQ:SH_periodic_coefficients} can be represented as $u(x,t) = u_{\rm Ans}(x,t) + \check{S}(x,t)$ where $\check{S}$ represents the error between the real solution and the approximation $u_{\rm Ans}(x,t)$. The PDE for the error $\check{S}$ is given by
\begin{equation}\label{eq:res_def}
\begin{split}
    \partial_t \check{S} = \mathcal{L}_p \check{S} -\check{S}^3 -3 \check{S}^2 u_{\rm Ans} - 3\check{S} u_{\rm Ans}^2-\text{Res}(u_{\rm Ans}) \, , \quad \text{Res}(u_{\rm Ans}) = -\partial_t u_{\rm Ans} + \mathcal{L}_pu_{\rm Ans} -u_{\rm Ans}^3 \, .
\end{split}
\end{equation}
The approximation $u_{\rm Ans}$ is constructed as a small amplitude solution that features $\varepsilon$ both in front of the amplitude and in the space and time scaling. A common technique is to make the residual small by matching the ansatz. However, to establish the validity of the approximation, it is important to control the entire error $\check{S}$.  Special situations are known in which the amplitude equation, although correctly derived, fails to make a correct prediction \cite{phd_thesis_failure_ampl, Failure_Nwave_Schneider_Haas, Failure_mod_eq_Schneider_Rijk_Haas, failure_slow_dyn_Schneider_BSZ}. This underscores the importance of validating the derived amplitude equation. The novelty of this paper lies in justifying the Ginzburg-Landau equation for a dissipative system with spatially varying coefficients that are not restricted in magnitude.

\begin{remark}[\bf Periodic solutions of GL]
The complex Ginzburg-Landau equation with constant coefficients,
\begin{equation}
    \partial_T A= \alpha_1 A + \alpha_2\partial_X^2A+\alpha_3|A|^2A\,,
\end{equation}
$\alpha_1, \alpha_2, \alpha_3\in \mathbb{C}$, admits stationary solutions that are spatially periodic of the form 
\begin{equation}
    A_{\rm per}(X,T) = Be^{i\kappa X}\,, \quad  |B|^2 = \frac{\kappa^2\alpha_2-\alpha_1}{\alpha_3}\, . 
\end{equation}
Upon taking $\alpha_1, ~\alpha_2$ and $\alpha_3$ to agree with \eqref{eq:GL_forced_SH} , these periodic solutions $A_{\rm per}$ lead to solutions of \eqref{EQ:SH_periodic_coefficients} featuring a band of possible wavenumbers close to $\ell_0$, given by 
\begin{equation}
    u_{\rm GL}(x,t) =  \varepsilon B\psi_1(\ell_0,x)e^{i(\ell_0+\varepsilon\kappa )x} +{\rm c.c.}.
\end{equation}
Solutions $A_{\rm per}$ are stable for a range of wavenumbers $\kappa$ (see e.g. \cite{Eckhaus1965, Cross_Hohenberg_1993, Schneider_Diffusive_stab, eckmann1997geometric, haragus2011local}). This family of periodic solutions of the Ginzburg-Landau equation explains the occurrence of an entire branch of stable periodic orbits for the Swift-Hohenberg equation. In more complicated systems, computing these solutions and determining their stability enables us, in particular, to compute the tip of the Busse balloon explicitly (see \cite{GKGS_SDHR}).
\end{remark}

\begin{remark}[\bf Global results]
For the classic constant-coefficient case the (above local) approximation result (based on \cite{SH_GL_val_cubic}) is complemented by an attractivity results and a global approximation result (see \cite{Bollerman1995},\cite{Eckhaus1993GLAttractor} , \cite{Schneider1994Global}, \cite{Schneider1995Analyticity}, \cite{Schneider2001}). We expect similar results to hold true for \eqref{EQ:SH_periodic_coefficients}, but leave this for future work. 
\end{remark}

\begin{remark}[$\bm{p(x) = p_*(x) + \varepsilon^2 \widetilde{p}(x)}$]\label{rem:p_tilde}
For $p(x) = p_*(x) + \varepsilon^2 \widetilde{p}(x)$ we expect a similar result with the modification that the GL equation will feature a term containing the change of the dispersion relation at $\ell = \ell_0$ in the direction $\widetilde{p}$. Extensions of the result of this form are left for future work.
\end{remark}

\begin{remark}[\bf $\bm{k_f}$ small]\label{rem:kf_small}
    From an ecological perspective, forcing with small wavenumbers, i.e. $k_f = o(1)$ with respect to $k_0$, are particularly interesting, as the wavenumber of mountains is much smaller than that of the observed wavenumber of vegetation patterns. From a mathematical point of view, forcing with small wavenumbers is of interest as for small forcing, the system admits space-intermittent solutions \cite{Doelman_Schielen}. Numerical experiments show that with large forcing, these space-intermittent solutions persist. Although often treated as a separate case, numerical study shows that for $k_f$ small and $\gamma$ large enough, the spectrum can satisfy Assumption~\ref{ASSUMP:Turing_instability} and hence the same results hold as for $k_f=\mathcal{O}(1)$. However, the numerical algorithm is more computationally expensive, hence numerical validation of Assumption~\ref{ASSUMP:Turing_instability} is less straightforward. Since the first two curves strongly overlap with small forcing, one has to carefully check Assumption~\ref{ASSUMP:Turing_instability}(f). 
\end{remark}

\subsection{Numerical Results}\label{sec:num_res}
In this section we present numerical approximations for a setting in which Assumption~\ref{ASSUMP:Turing_instability} appears to hold in the case $p(x) = p_*(x) + \varepsilon^2 r = \mu_* + \gamma_* \cos(k_f x)+ \varepsilon^2 r$. For parameter values right after onset, a perturbation to the zero state $u_*=0$ leads to a new pattern solution. These patterns are approximated for $k_f$ non-resonant as in Theorem \ref{thm:SH_approx}. If $A_1$ is real, a stationary solution of (\ref{eq:GL_forced_SH}) is the constant solution
\begin{equation}
    A_1 = \pm \sqrt{\frac{r}{3\int_0^{\frac{2\pi}{k_f}}|\psi_1(\ell_0,x)|^4\rho(x)~dx}}.
\end{equation}

 Following the procedure described in Section~\ref{SEC:Floquet-Bloch}, we can compute $\psi_1$ numerically. Using the constant solution for $A_1$, we find the approximations for the converged pattern solutions of \eqref{EQ:SH_periodic_coefficients} as shown in Figure \ref{fig:SH_approx_nonres}. This confirms that \eqref{eq:u_GL} provides a very good approximation for the solutions of \eqref{EQ:SH_periodic_coefficients}.

\begin{figure}[ht]
\centering
    \begin{subfigure}[b]{0.4\textwidth}
    \includegraphics[width=\textwidth]{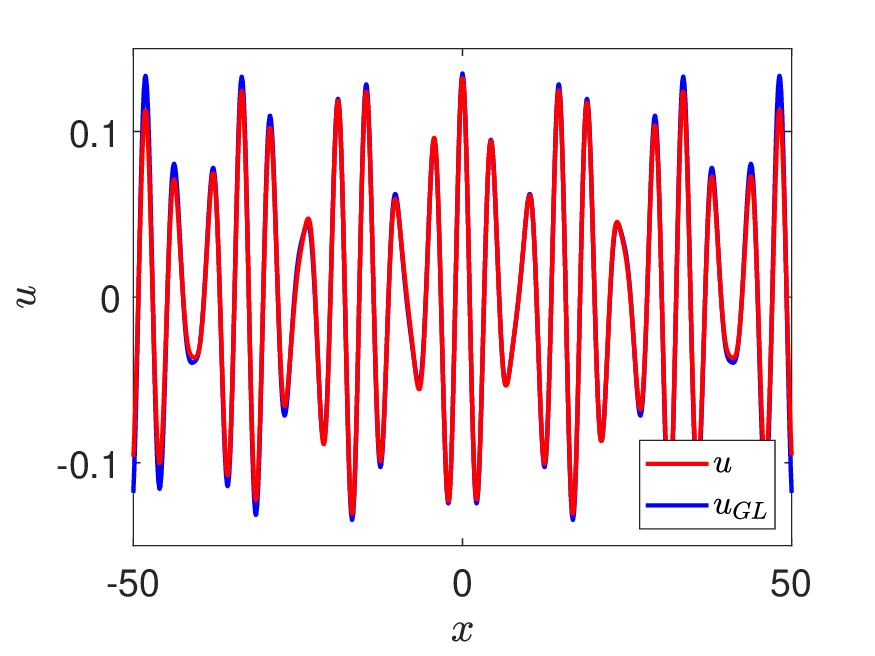}
    \caption{$k_f=3$}
  \end{subfigure}
    \hspace{.5cm}
    \begin{subfigure}[b]{0.4\textwidth}
    \includegraphics[width=\textwidth]{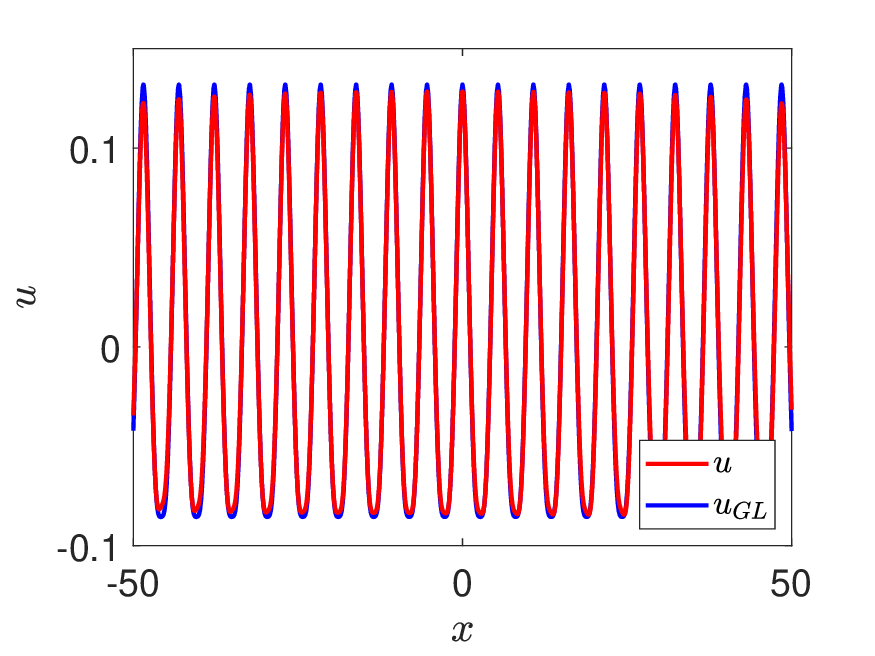}
    \caption{$k_f=3.5$}
  \end{subfigure}
  \\
      \begin{subfigure}[b]{0.4\textwidth}
    \includegraphics[width=\textwidth]{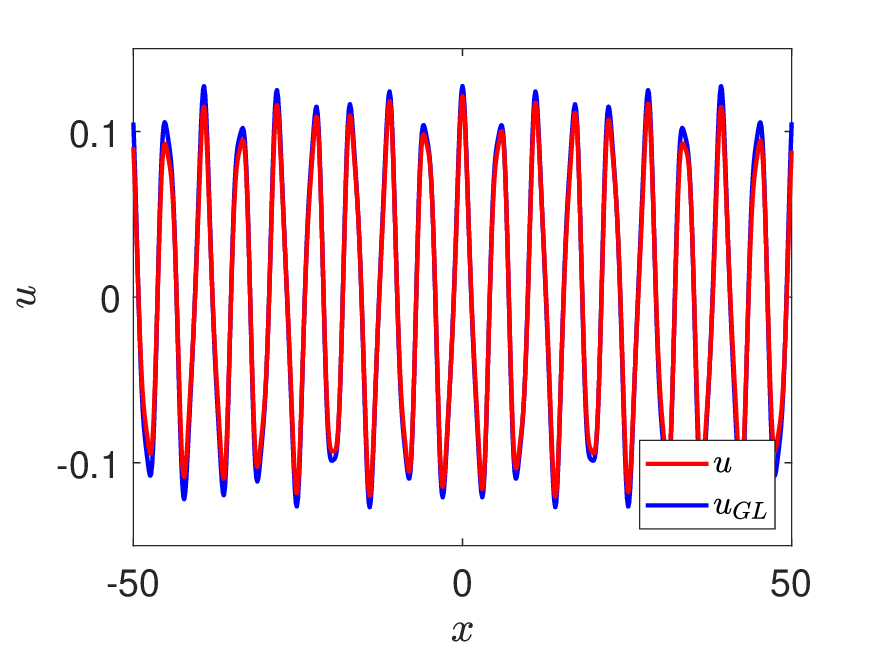}
    \caption{$k_f=4$}
  \end{subfigure}
    \hspace{.5cm}
    \begin{subfigure}[b]{0.4\textwidth}
    \includegraphics[width=\textwidth]{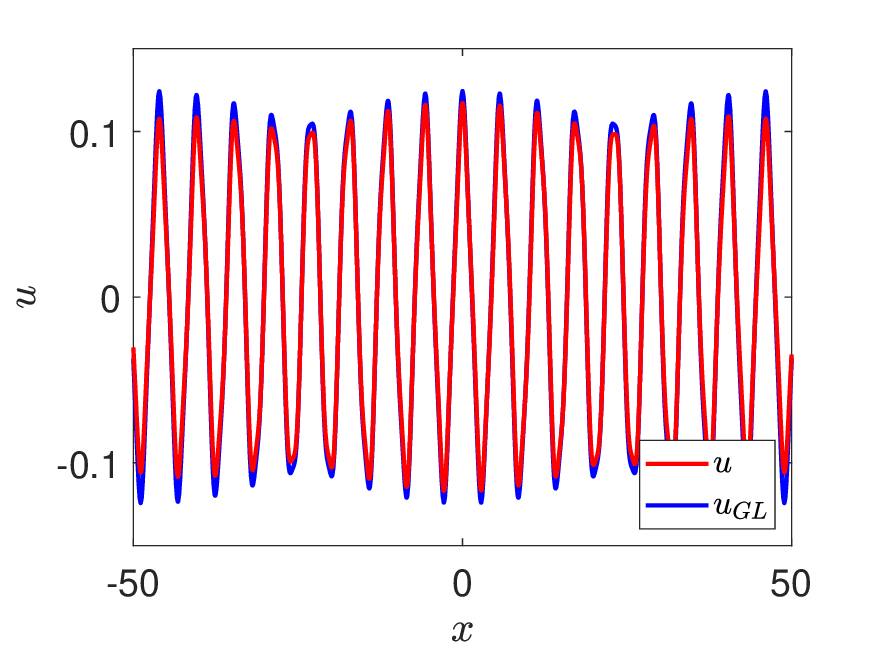}
    \caption{$k_f=4.5$}
  \end{subfigure}
\caption{Stationary solution $u$ of SH in red and approximation $u_{\rm GL}$ in blue, for $p_*= \mu_*+\gamma_*\cos(k_fx)$, ~$\rho(x)=1$, with different non-resonant values of $k_f$. For all figures, we have $\varepsilon=0.1$, $k_0=1$, $\mu_*=-0.76$, $\gamma=\gamma_*(k_f, \mu_*)$ and $r=1$.}
\label{fig:SH_approx_nonres}
\end{figure}

 For wavenumbers where $\ell_0=0$, a constant solution of \eqref{eq:AC_forced_SH} is given by
 \begin{equation}
     A_1=\pm \sqrt{\frac{r}{\int_0^{\frac{2\pi}{k_f}}\psi_1(\ell_0,x)^4\rho(x)~dx}},
 \end{equation}
 while for resonant wavenumbers $k_f$ and $A_1$ real, a constant solutions of (\ref{eq:GL_forced_SH_res}) is
\begin{equation}
     \quad A_1=\pm \sqrt{\frac{r}{3\int_0^{\frac{2\pi}{k_f}}|\psi_1(\ell_0,x)|^4\rho(x)~dx+\int_0^{\frac{2\pi}{k_f}}\psi_1(\ell_0,x)^3\overline{\psi_1(\ell_0,x)}\rho(x)~dx}}\, .
\end{equation}
Numerical results for these cases are shown in Figure \ref{fig:SH_approx_kfk0_kf2k0}, confirming Corollary \ref{cor:SH_approx_l00} and \ref{cor:SH_approx_res}.

    \begin{figure}[t]
\centering
    \begin{subfigure}[]{0.4\textwidth}
    \includegraphics[width=\textwidth]{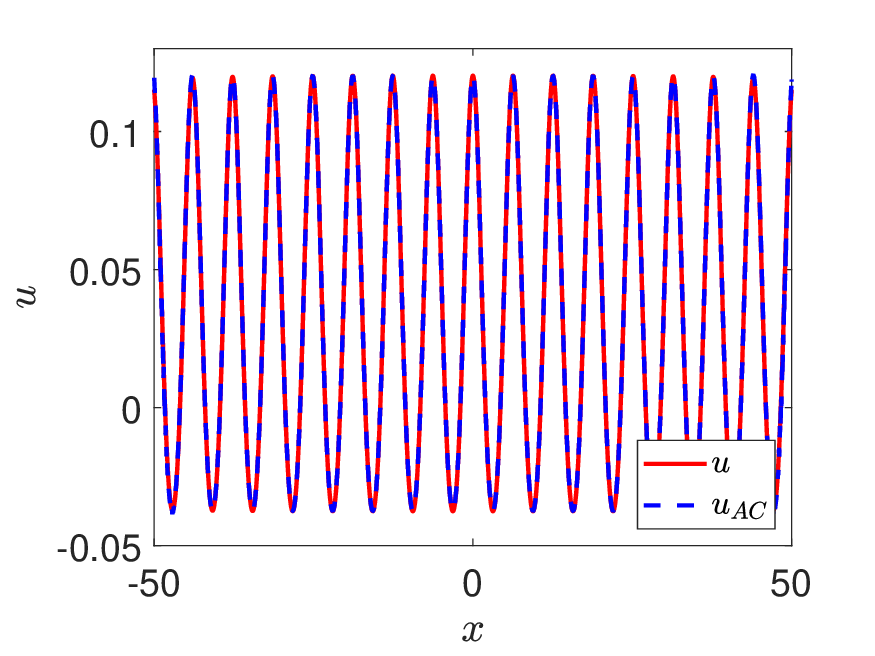}
    \caption{$k_f=1$}
  \end{subfigure}
    \begin{subfigure}[]{0.4\textwidth}
    \includegraphics[width=\textwidth]{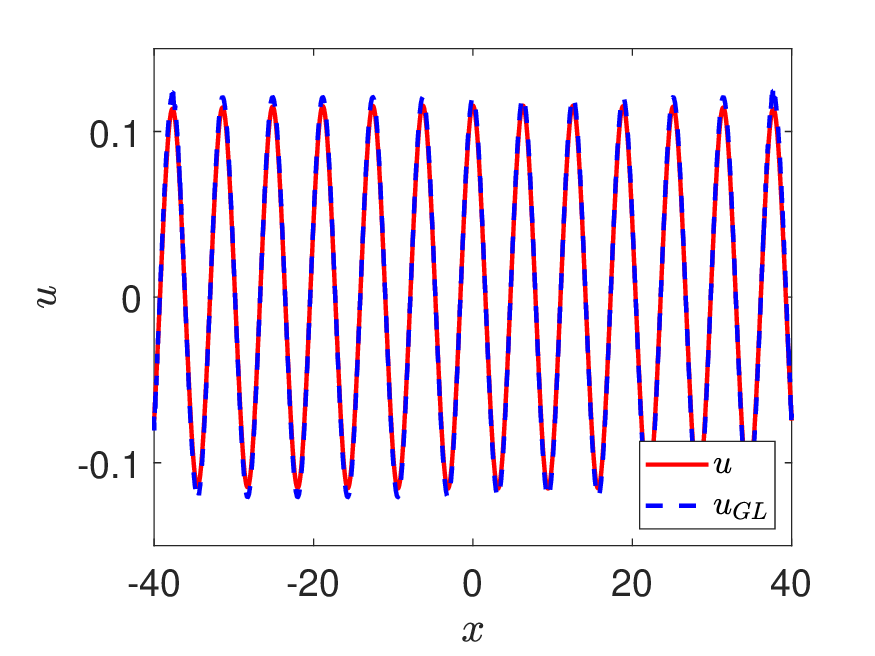}
    \caption{$k_f=2$}
  \end{subfigure}
\caption{Stationary solution $u$ of SH in red and approximation $u_{\rm AC}$ or $u_{\rm GL}$ in blue, for $p_*= \mu_*+\gamma_*\cos(k_f x)$, $\rho(x)=1$, with different resonant values $k_f$. For all figures, we have $\varepsilon=0.1$, $k_0=1$, $\mu_*=-0.76$, $\gamma=\gamma_*(k_f, \mu_*)$ and $r=1$.}
\label{fig:SH_approx_kfk0_kf2k0}
\end{figure}

\subsection{Plan of the paper} The rest of this article is organized as follows. In Section~\ref{SEC:derivation} we discuss the linear analysis and corresponding numerical explorations of the band structure in more detail and give the derivation of the GL equation in Bloch space. In Section~\ref{sec:justification} we give a rigorous estimate for the residual followed by a rigorous estimate of the approximation error which proves our main result, Theorem~\ref{thm:SH_approx}. We close with an overview of related work in Section~\ref{sec:related_work} and a discussion and outlook in Section~\ref{sec:discussion}.

\newpage
\subsection{Notation}
We follow the notation and definitions in \cite{SchneiderUecker}.

\subsubsection{Bloch transform}
The Bloch transform is an adaption of the Fourier transform in the spatially homogeneous case to the spatially periodic case \cite{Simon_Reed}. It is defined using the Fourier transform and, similarly as to the Fourier transform, multiple definitions can be used. We use the symmetric Fourier transform
\begin{equation}
    \hat{u}(k)=(\mathcal{F}u)(k) =\frac{1}{\sqrt{2\pi}} \int_{-\infty}^\infty u(x)e^{-ikx} ~dx,
\end{equation}
with inverse Fourier transform
\begin{equation}
    u(x)=(\mathcal{F}^{-1}\hat{u})(k)=\frac{1}{\sqrt{2\pi}} \int_{-\infty}^\infty \hat{u}(k) e^{ikx}~dk.
\end{equation}
The following definition for the Bloch transform for $2\pi/k_f$-periodic functions can be derived
\begin{equation}\label{eq:Bloch_T}
    \widetilde{u}(\ell,x)=(\mathcal{B}u)(\ell, x) = \sum_{j\in\mathbb{Z}} \hat{u}\left(k_f j+\ell\right) e^{ik_f jx},
\end{equation}
with inverse Bloch transform
\begin{equation}\label{eq:Bloch_inv_T}
    u(x) = (\mathcal{B}^{-1}\widetilde{u})(x) = \frac{1}{\sqrt{2\pi}} \int_{-\frac{k_f}{2}}^\frac{k_f}{2}\widetilde{u}(\ell,x) e^{i\ell x} ~d\ell.
\end{equation}
This Bloch transform satisfies the following properties
\begin{equation}
\begin{split}
    \widetilde{u}\left(\ell,x+\frac{2\pi}{k_f}\right) &= \widetilde{u}(\ell,x),\\
    \widetilde{u}\left(\ell+k_f,x\right) &= \widetilde{u}(\ell, x)e^{-ik_f x}.
    \end{split}
\end{equation}
The multiplication to convolution property becomes 
\begin{equation}\label{eq:Bloch_conv_def}
    \widetilde{fg}(\ell,x) = \int_{-k_f/2}^{k_f/2} \widetilde{f}(p) \widetilde{g}(\ell - p) ~dp = (\widetilde{f} \star \widetilde{g})(\ell,x).    
\end{equation}

\subsubsection{Function space}\label{subsec:function_spaces}

Following the reasoning in \cite{SchneiderUecker}, we want to allow solutions of the Ginzburg-Landau equation that do not decay for $|x|\rightarrow \infty$ and want to be able to use the Fourier transform and the Bloch transform. Hence we consider the Sobolev space $H_{ul}^\vartheta$, which is constructed in the following way.

\medskip

We take a positive weight function $q: \mathbb{R}\rightarrow (0,\infty)$ that is continuous, bounded and satisfies $\int_\mathbb{R} q(x) ~dx <\infty$. Furthermore, we impose $q\in C^2(\mathbb{R}, \mathbb{R})$ and $|q'(x)|$, $|q''(x)|\leq q(x)$ for all $x$. Introduce
\begin{equation}
    \widetilde{L_{ul}^2} = \left\{ u\in L_{loc}^2(\mathbb{R})~:~\|u\|_{L_{ul}^2} <\infty\right\}, \quad \text{ where } \quad \|u\|_{L_{ul}^2}^2 = \sup_{y\in\mathbb{R}} \int_\mathbb{R} q(y+x) u(x)^2~dx.
\end{equation}
The $L_{ul}^2$ norm is equivalent to the norm $\|\cdot\|_*$ defined as 
\begin{equation}
    \|u\|_* = \sup_{y\in\mathbb{R}} \left( \int_{y-1/2}^{y+1/2} |u(x)|^2 ~dx \right)^{1/2}.
\end{equation}
Introduce the translation operator $T_y~:~\widetilde{L_{ul}^2} \rightarrow \widetilde{L_{ul}^2}$, $(T_y u)(\cdot) \mapsto u(\cdot + y)$. The space of uniformly local $L^2$ functions is given by
\begin{equation}
    L_{ul}^2 = \left\{ u\in\widetilde{L_{ul}^2}~:~ \|T_y u - u\|_{L_{ul}^2} \rightarrow 0 \text{ as } y\rightarrow 0\right\}.  
\end{equation}
For $\vartheta\in\mathbb{N}$ we define the uniformly local Sobolev space $H_{ul}^\vartheta$ as
\begin{equation}
    H_{ul}^\vartheta = \left\{ u\in L_{ul}^2 ~:~ D^m u \in L_{ul}^2 \text{ for all } |m|<\vartheta \right\}.
\end{equation}

\newpage

\section{Derivation of the Ginzburg-Landau equation}\label{SEC:derivation}
The conventional way of deriving the Ginzburg-Landau equation is based on making an ``ansatz" for the solution. While this is sometimes presented as an educated guess based on observations, a dynamical systems viewpoint reveals that it is feasible to craft this ansatz along the spectral properties of the linearized PDE during a bifurcation. This section is, hence, structured as follows.

\medskip

{\bf Plan of the section.} We first briefly discuss the origin of the band structure for periodic problems in a way that will allow to give a numerical procedure to compute it in Section~\ref{SEC:Floquet-Bloch}. We will complement this by another viewpoint that directly gives a characterization of coefficients that give rise to an ONB, so that (a)-(c) in Assumption~\ref{ASSUMP:Turing_instability} are fulfilled. Motivated by the knowledge of the shape of solutions for the linearized equation we derive the Bloch transformed, diagonalized system which explicitly features the bands $\lambda_n$ in Section~\ref{SEC:GL_derivation}. This then allows to tailor an elaborate ansatz in Bloch space which, in particular, yields the (Bloch transformed) Ginzburg-Landau equation and matches the claimed ansatz in Theorem~\ref{thm:SH_approx}.

\subsection{Linear analysis via Floquet-Bloch theory}\label{SEC:Floquet-Bloch}
Consider the eigenvalue problem for the linear operator arising from linearizing the Swift-Hohenberg equation \eqref{EQ:SH_periodic_coefficients} around $u_* = 0$,
\begin{align}\label{EQ:EVP}
    \mathcal{L}_p v(x) = \lambda v(x) \,,
\end{align}
with $\mathcal{L}_p$ as in \eqref{EQ:linear_operator}. It is well-known (see, e.g. \cite{Eastham_old}) that this operator is closed and self-adjoint on $L^2(\mathbb{R})$ and that it only has essential spectrum in the form of a countably-infinite number of closed intervals that may or may not overlap. This structure can be deduced in various ways of which we will now briefly outline two.

\medskip

By Floquet theory, bounded fundamental solutions of \eqref{EQ:EVP} can be chosen of the form
\begin{align}\label{EQ:Floquet_FS}
    v_+(x, \lambda) = e^{i\ell(\lambda)x}\psi(\ell(\lambda), x) \,, \quad v_-(x, \lambda) = e^{-i\ell(\lambda)x}\psi(-\ell(\lambda), x)  \, , \quad  0 \leq \ell < k_f/2 \, ,
\end{align}
where $\psi(-\ell(\lambda), x) = \overline{\psi(\ell(\lambda), x)}$ and $\psi\left(\ell, x + \frac{2\pi}{k_f}\right) = \psi(\ell, x), x \in \mathbb{R}$. There are either none, one or two such pairs since \eqref{EQ:EVP} is a fourth order ODE. In other words, $\lambda \in \sigma_{\rm ess}(\mathcal{L}_p)$ if and only if there are purely imaginary Floquet exponents $\widetilde{\ell}(\lambda) = i \ell(\lambda) \in i \mathbb{R}$. Floquet theory gives explicit formulas for the Floquet exponents in terms of the canonical fundamental matrix $\Phi = \Phi(\lambda)$. Consequently, if $\Phi$ is known explicitly, so is the band structure $(\ell_{m(\lambda)}(\lambda), \lambda)$ with $m(\lambda) \in \{0, 2, 4\}$, that is, for any fixed $\lambda \in \mathbb{R}$ there are either $0, 2$ or $4$ purely imaginary Floquet exponents. Numerical implementation of these relations allow to explore the possible band structure, whose infinitely many bands can be deduced from the Floquet discriminant (see, e.g. \cite{Eastham_old}).

\medskip

The origin of the band structure sketched in Figure~\ref{fig:SH_spec_cartoon} can be predicted analytically by adopting the following viewpoint. Using \eqref{EQ:Floquet_FS} as coordinate transformation for the eigenvalue problem \eqref{EQ:EVP} gives
\begin{align}\label{EQ:EVP_Bloch}
    \widetilde{\mathcal{L}}_p \psi = \lambda \psi \, , \quad \psi\left(x + \frac{2\pi}{k_f}\right) = \psi(x) \, ,
\end{align}
for
\begin{align}\label{EQ:L_p_tilde}
    \widetilde{\mathcal{L}}_p : = -[(\partial_x + il)^2 + k_0^2]^2 + p \, .
\end{align}
For $p(x) = \varepsilon^2 r, r \in \mathbb{R},$ we have $\widetilde{\mathcal{L}}_0 : = r-[(\partial_x + il)^2 + k_0^2]^2$ for which $\psi_n(x) = \sqrt{\frac{k_f}{2\pi}} e^{\pm i k_fnx}, n \in \mathbb{N},$ and $\lambda_n(\ell) = -[k_0^2-(\pm k_f n + \ell)^2]^2+\varepsilon^2 r$ (see blue curves in Figure~\ref{fig:SH_def_spec}). Noting that $\widetilde{\mathcal{L}}_p$ can be viewed as a perturbation of $\widetilde{\mathcal{L}}_0$ directly gives the following (see, e.g. \cite{Schneider_Diffusive_stab}).

\begin{lemma}\label{LEMMA:bandstructure}
Let $p$ be a continuous periodic function on $[0,2\pi/k_f]$. Then for any fixed $\ell$ we have that $\widetilde{\mathcal{L}}_p$ from \eqref{EQ:L_p_tilde} is a closed, self-adjoint linear operator on $\chi:= L^2_{\rm per}([0,2\pi/k_f];\mathbb{C})$ with eigenvalues $\lambda_n(\ell) \in \mathbb{R}, n \in \mathbb{N},$ and corresponding eigenfunctions $\{\psi_n(\ell, \cdot)\}_{n \in \mathbb{N}}$ that form an orthonormal basis of $\chi$. Furthermore, $\lambda(-\ell) = \lambda(\ell)$ and $\psi(-\ell, \cdot) = \overline{\psi(\ell, \cdot)}$. In particular, a periodic extension of $p$ fulfills (a)-(c) in Assumption~\ref{ASSUMP:Turing_instability}.
\end{lemma}

\begin{figure}[ht]
    \begin{subfigure}[]{0.32\textwidth}
    \includegraphics[width=\textwidth]{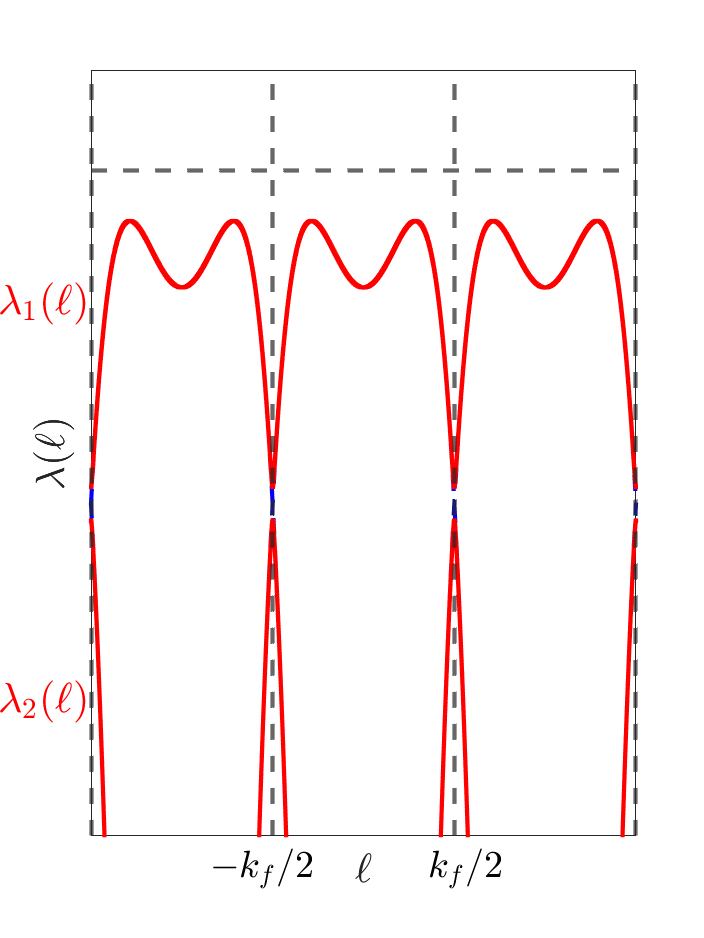}
    \caption{Spectrum for $\gamma=5\varepsilon$}
  \end{subfigure}
    \begin{subfigure}[]{0.32\textwidth}
    \includegraphics[width=\textwidth]{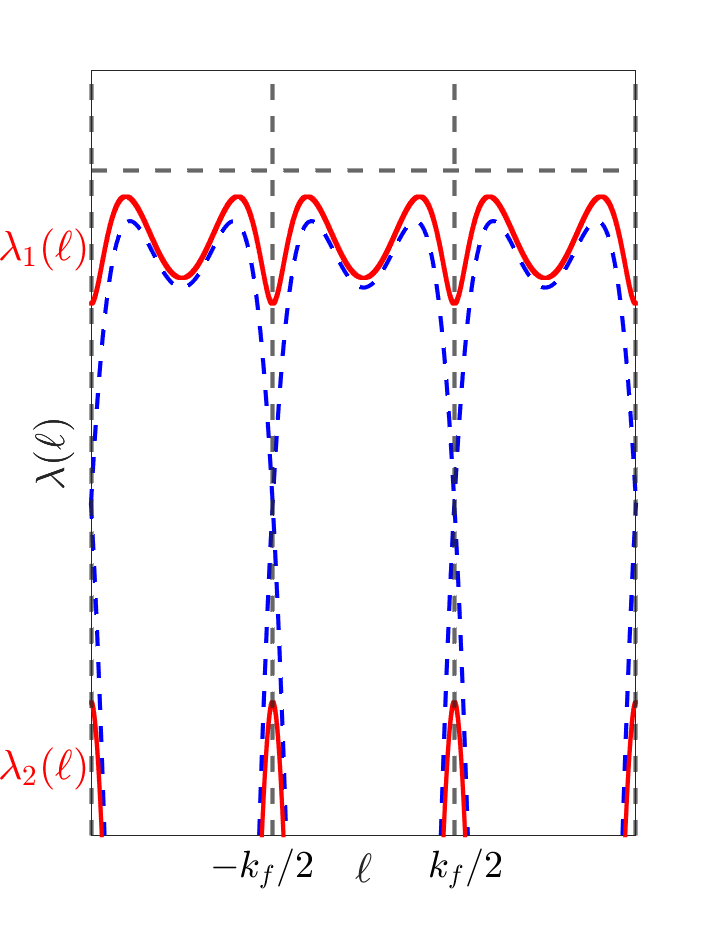}
    \caption{Spectrum for $\gamma=\mathcal{O}(1)$, $\gamma<\gamma_*$}
  \end{subfigure}
    \begin{subfigure}[]{0.32\textwidth}
    \includegraphics[width=\textwidth]{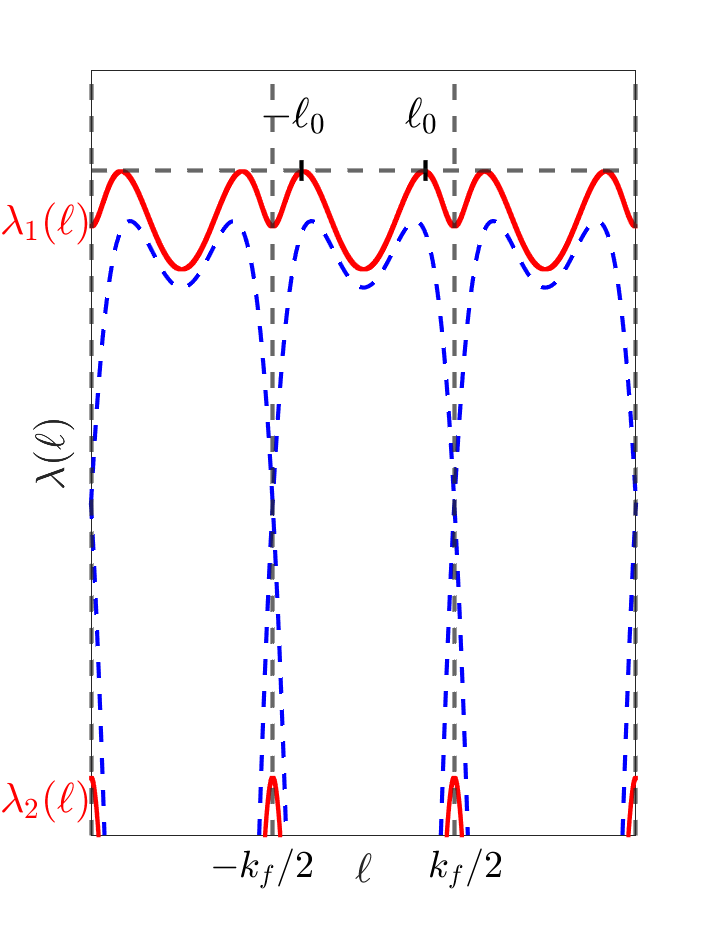}
    \caption{Spectrum for $\gamma=\gamma_*$}
  \end{subfigure}
\caption{Deformation of the spectral curves for the linear periodically forced Swift-Hohenberg system upon setting $p(x) = \mu_* + \gamma \cos(k_f x)$ and increasing $\gamma$. Parameter values are $\mu_* = -0.76$, $k_0=1$ and $k_f=3.5$. Blue dashed lines are artificial spectral curves for $\gamma=0$, the red lines are the spectral curves for $\gamma>0$ for different values of $\gamma$. The value $\gamma_*$ is such that the first band yields bifurcation onset.}
\label{fig:SH_def_spec}
\end{figure}

\subsubsection{Numerical computations of the band structure}
While Lemma~\ref{LEMMA:bandstructure} suggests that the occurrence of a band structure is rather generic, finding $p$ that generates a Turing bifurcation is rather difficult analytically. Nevertheless, there is a numerical recipe that allows to determine $p$ that fulfills Assumption~\ref{ASSUMP:Turing_instability}. This recipe is based on the knowledge of how the deformation of the constant-coefficient case band structure takes place upon adding periodic coefficients. For all the numerical results, we use the form
\begin{equation}
    p(x) = p_*(x)+\varepsilon^2r= \mu_* +\gamma_*\cos(k_fx)+\varepsilon^2r.
\end{equation}We illustrate in Figure~\ref{fig:SH_def_spec} the following procedure:
\begin{itemize}
    \item First we set $p(x) = \mu_*$ (artificially viewing this as a periodic function with period $2\pi/k_f$) for some $\mu_* < 0$, which, in particular, already sets up the first band $\lambda_1(\ell)$ to have a shape leading to a Turing bifurcation.
    \item Then we tune $\gamma_*$ such that $p(x) = p_*(x) = \mu_* + \gamma_* \cos(k_f x)$ fulfills (e)-(f) in Assumption~\ref{ASSUMP:Turing_instability}, so that it places the system at the onset of a bifurcation.
    \item Finally, we set $p(x) = p_*(x) + \varepsilon^2 r = \mu_* + \gamma_* \cos(k_f x)+ \varepsilon^2 r$, which creates a bifurcation upon merely shifting the band structure up and down via $r \in \mathbb{R}$ (cf. Figure~\ref{fig:SH_spec_bif_per}), fulfilling (d) in Assumption~\ref{ASSUMP:Turing_instability}.
\end{itemize}
Note that this further confirms, at least numerically, that Assumption~\ref{ASSUMP:Turing_instability} can hold true. For the computations displayed in the following, we set $k_0=1, k_f=3.5$ and vary $r$. The first spectral curve around bifurcation onset is shown in Figure \ref{fig:SH_spec_bif_per}.

\begin{figure}[ht]
\centering
    \begin{subfigure}[]{0.32\textwidth}
    \includegraphics[width=\textwidth]{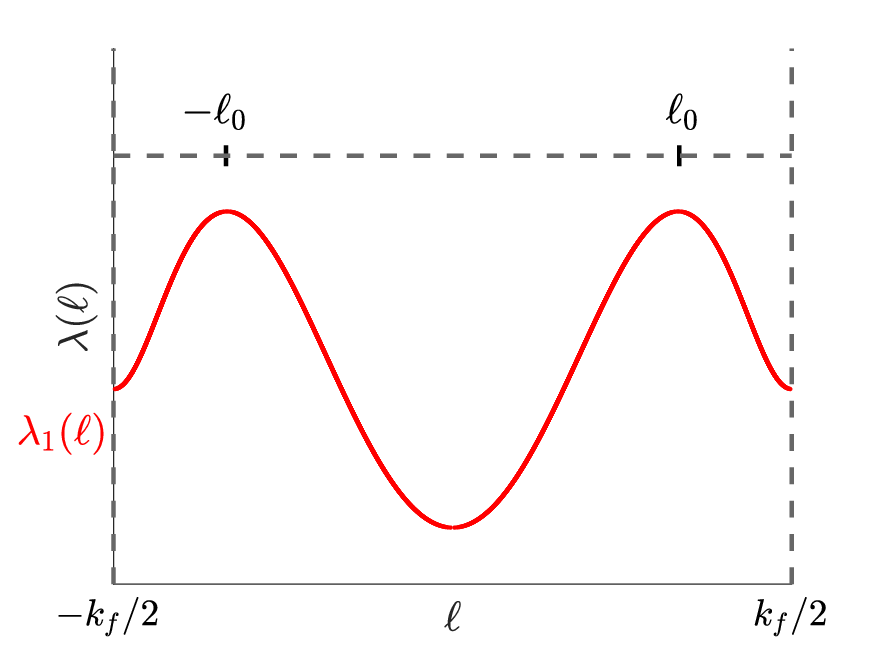}
    \caption{Spectrum for before onset}
  \end{subfigure}
    \begin{subfigure}[]{0.32\textwidth}
    \includegraphics[width=\textwidth]{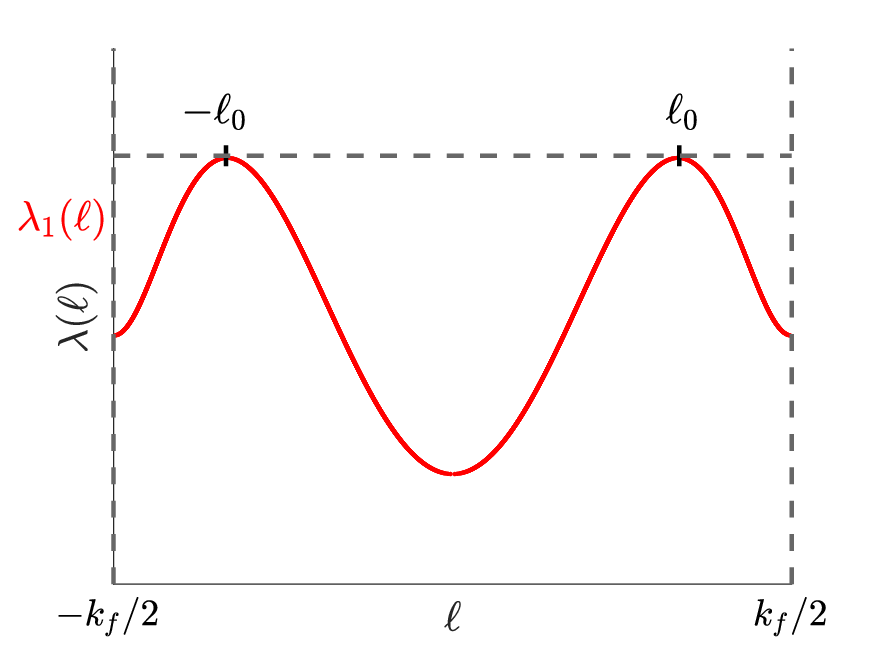}
    \caption{Spectrum at onset}
  \end{subfigure}
    \begin{subfigure}[]{0.32\textwidth}
    \includegraphics[width=\textwidth]{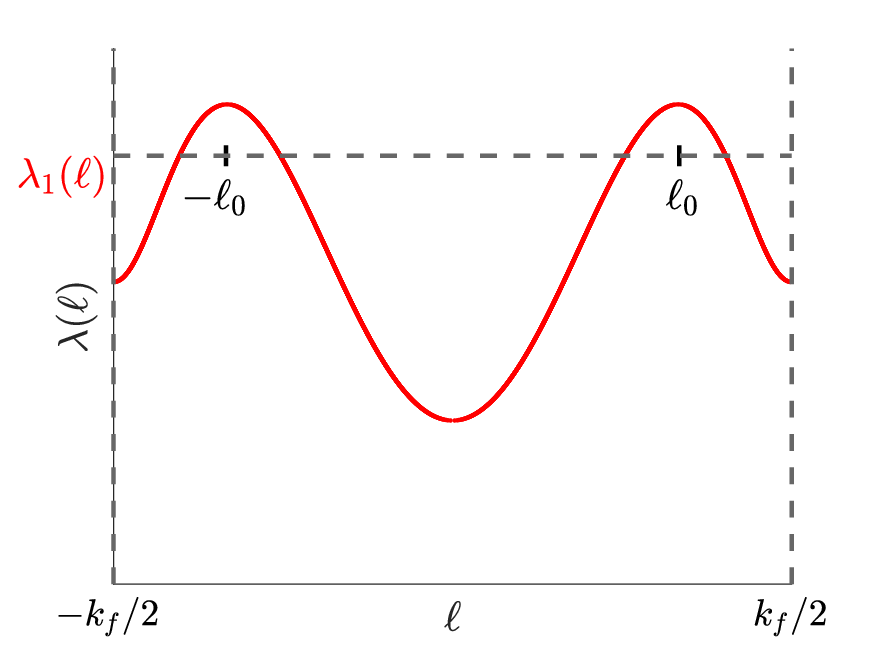}
    \caption{Spectrum after onset}
  \end{subfigure}
\caption{Spectrum when setting $p(x) = \mu_* + \gamma_* \cos(k_f x) + \varepsilon^2 r$ around before onset, $r = -1$ (left), at onset, $r = 0$, (middle) and after bifurcation, $r = 1$ (right). Other parameter values are $k_0=1$, $k_f=3.5$ and $\mu_*=-0.76$.}
\label{fig:SH_spec_bif_per}
\end{figure}

\begin{remark}(Finding $\mu_*$ and $\gamma_*$)
    Note that finding a $\mu_*$ and $\gamma_*$ such that Assumption~\ref{ASSUMP:Turing_instability}(d)-(f) hold true might not always be trivial. Once an initial guess for $\mu_*$ is made, an educated guess for $\gamma_*$ can be found by continuation of the background solution in the $\gamma$ direction. Using the Matlab package pde2path \cite{pde2path_HU, pde2path_HU_DW_JR} it is possible to find the first bifurcation point, which serves as an initial guess for (and is often very close to) $\gamma_*$. The value of $\mu_*$ has to be adapted if it turns out that at the found $\gamma_*$, Assumption~\ref{ASSUMP:Turing_instability}(f) is not satisfied yet.
\end{remark}

\begin{figure}[ht]
\centering
    \begin{subfigure}[b]{0.37\textwidth}
    \includegraphics[width=\textwidth]{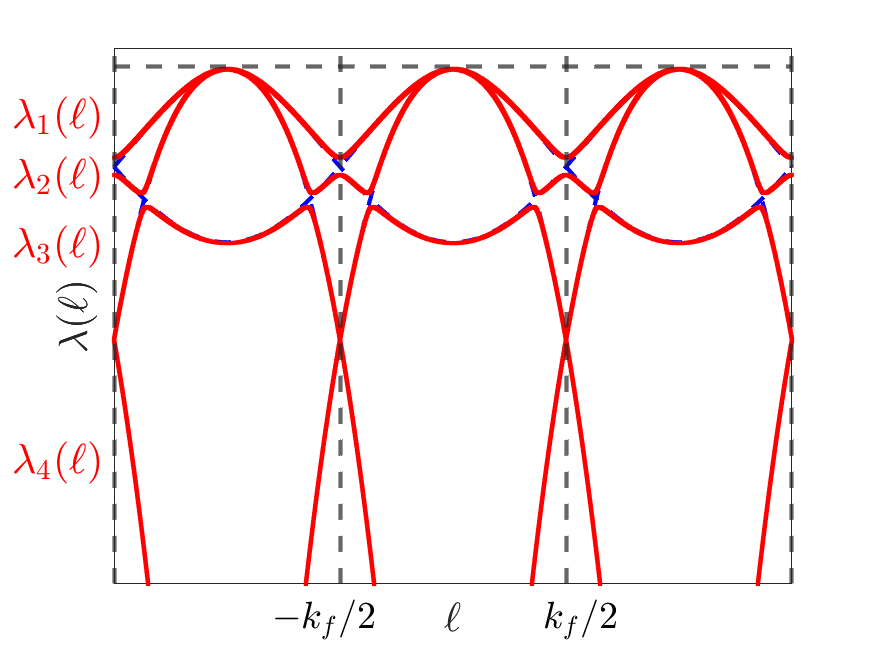}
    \caption{Spectrum at onset for $\gamma=\varepsilon$, $k_f=k_0$}
  \end{subfigure}
    \begin{subfigure}[b]{0.37\textwidth}
    \includegraphics[width=\textwidth]{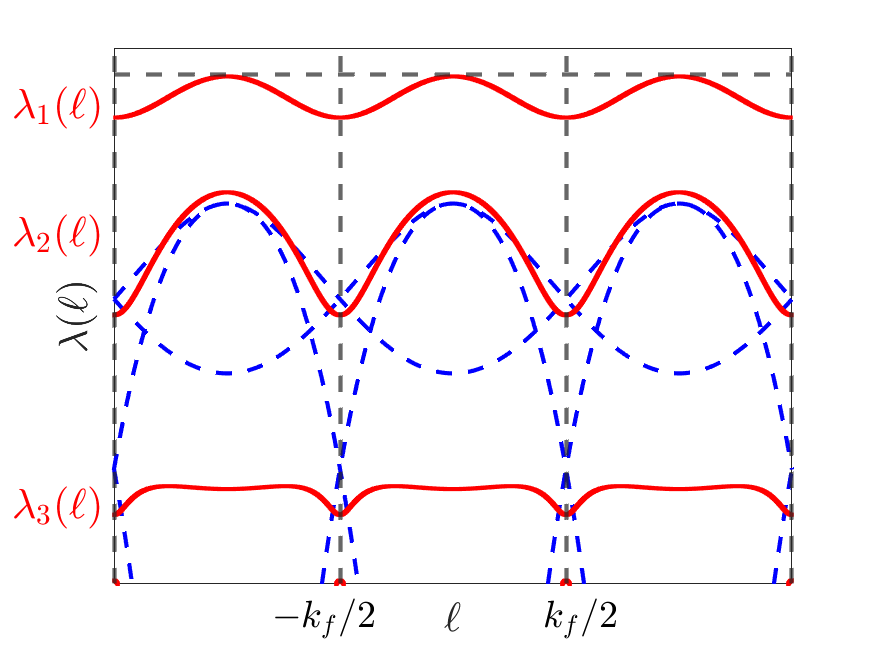}
    \caption{Spectrum at onset for $\gamma=\mathcal{O}(1)$, $k_f=k_0$}
  \end{subfigure}
  
    \begin{subfigure}[b]{0.37\textwidth}
    \includegraphics[width=\textwidth]{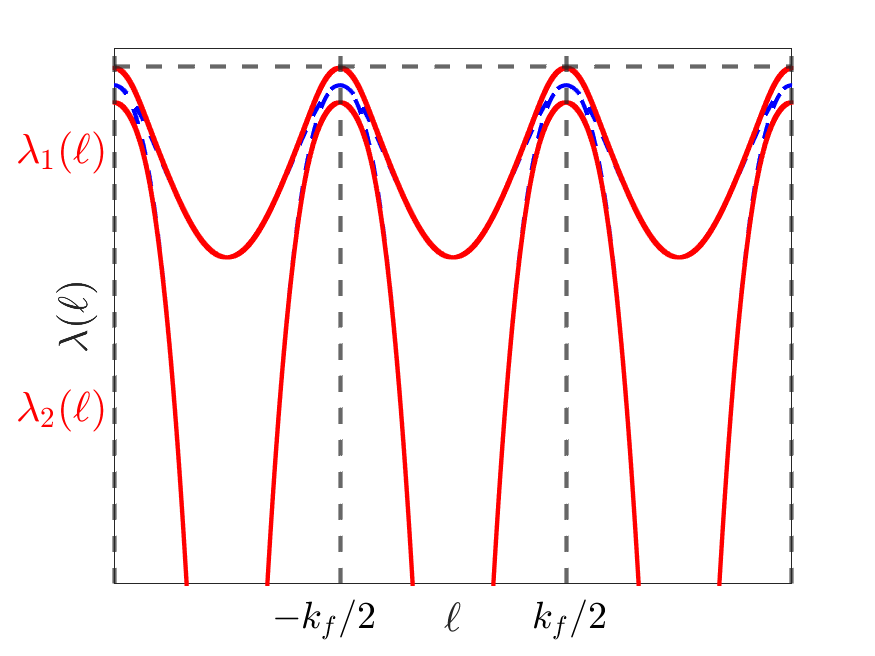}
    \caption{Spectrum at onset for $\gamma=\varepsilon$, $k_f=2k_0$}
  \end{subfigure}
    \begin{subfigure}[b]{0.37\textwidth}
    \includegraphics[width=\textwidth]{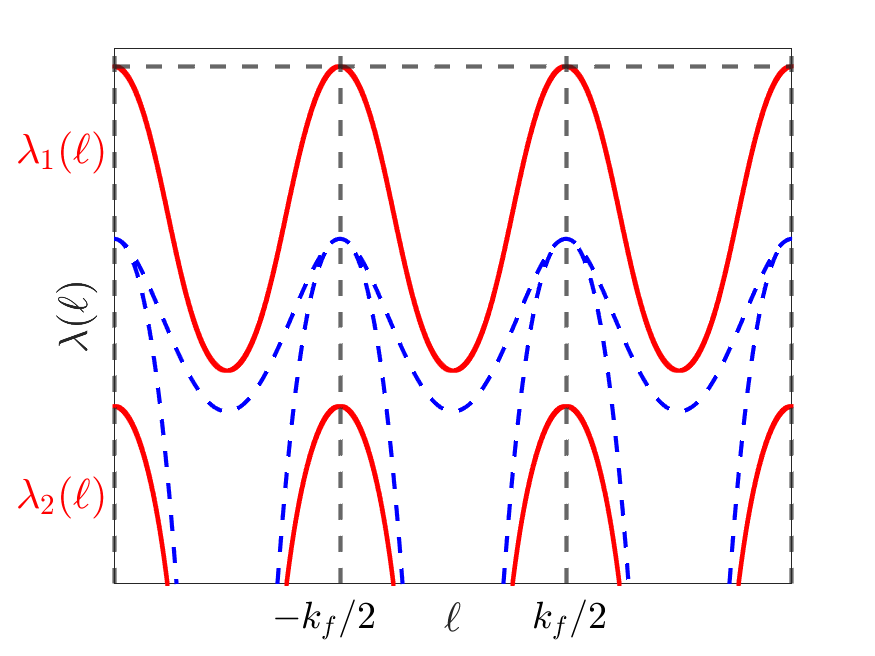}
    \caption{Spectrum at onset for $\gamma=\mathcal{O}(1)$, $k_f=2k_0$}
  \end{subfigure}
\caption{Spectrum of periodically forced Swift-Hohenberg system with forcing strength $\gamma= \mathcal{O}(\varepsilon)$ (left) or $\gamma=\mathcal{O}(1)$ (right) for resonant values $k_f=k_0$ (upper), representing the general $N$ is even case, and $k_f=2k_0$ (lower), for the $N$ is odd case. The dashed blue lines are the spectrum without forcing, red the spectral curves for $\gamma=\gamma_*>0$.}
\label{fig:SH_spec_res_cases}
\end{figure}

\subsubsection{Resonance phenomena}\label{sec:resonances}
The interplay between the intrinsic wavenumber $k_0$ and the forcing wavenumber $k_f$ can lead to resonance phenomena. These phenomena are studied in case of small forcing in \cite{Ehud_2008, Ehud_2014, Ehud_2012, Ehud_2015_1, Ehud_2015_2, Ehud_spatialperforcing, uecker2001, Doelman_Schielen}. The specific choice of $k_f$ relative to $k_0$ can lead to different behaviors in the system. The distinction between resonant and non-resonant cases becomes evident when studying the spectrum of the linear operator. We now give several cases of resonances, although we will focus for our analytical results on the case when the forcing is sufficiently large to have led to a spectral gap between the first two spectral curves.

{\bf Range of $\bm{\gamma}$.} From numerical observations, we know that for every value of $k_f$, there exists a $\gamma'\gg 1$ such that for any $(\mu_*, \gamma_*)$ with $\gamma_*\geq \gamma'$ the first spectral curve $\lambda_1(\ell)$ has deformed such that the maxima are located at $\pm \ell_0=\pm \frac{k_f}{2}$. 
Due to the periodicity of the spectrum, this means that $\ell_0=3\ell_0\pmod{k_f}$. Hence in the presence of for sufficiently large forcing, any wavenumber $k_f$ can be resonant.

\medskip

{\bf No spectral gap.} In case of small forcing, $\gamma =\mathcal{O}(\varepsilon)$, there are ranges of wavenumbers $k_f$ for which there is no spectral gap at onset between $\lambda_1$ and $\lambda_2$. In this case, an elaborate examination of the maximum of $\lambda_2$ is required. Denote the location of the maximum of $\lambda_2(\ell)$ by $\ell_1$. If $\lambda_2(\ell_1)=\mathcal{O}(\varepsilon^2)$, both $\lambda_1$ and $\lambda_2$ can cross the zero axis after bifurcation onset, yielding two instabilities directly after each other. The ansatz would have to also account for this second instability. On the other hand, if $\lambda_2(\ell_1)=\mathcal{O}(\varepsilon)$, only $\lambda_1$ crosses zero when increasing $\mu$ to $\mu_*+r\varepsilon^2$, hence an approximation can be derived using the analogous ansatz. However, in addition to the violation of Assumption~\ref{ASSUMP:Turing_instability}(f), Assumption~\ref{ASSUMP:Turing_instability}(g) might not be satisfied in this case. We leave it to later work to explore the nonlinear stability for these cases. 

\medskip

{\bf Large wavenumber forcing.} Even though a priori we did not put assumptions on the magnitude of $k_f$, we do exclude the cases $k_f\gg 1$. From a spectral point of view, $k_f\gg 1$ means that the period of the spectrum becomes large, meaning that the intersections of the first two spectral curves occur at very low values of $\lambda_1(\ell)$. As a consequence, an extremely large value of $\gamma$ is needed to see changes in the spectrum and to destabilize the background state. The case $k_f\gg 1$ is therefore part of future research.

\medskip

{\bf Resonant wavenumbers.} In case of small forcing $\gamma=\mathcal{O}(\varepsilon)$, resonance phenomena can occur for specific values of $k_f$ \cite{Ehud_2008, Ehud_2014, Ehud_2012, Ehud_2015_1, Ehud_2015_2, Ehud_spatialperforcing, uecker2001, Doelman_Schielen}. These resonant wavenumbers are $k_f\approx \frac{2k_0}{N}$ for $N>0$ \cite{Doelman_Schielen}. In the spectral picture, this means that, due to the $k_f$-periodicity, the artificial curves for $\gamma=0$ intersect at the maxima. For $N$ even, the first spectral curve has its maximum at $\ell_0=0$, while for $N$ odd, the spectral pictures show that $\ell_0 = \frac{k_f}{2}$. 
In case of small forcing, there is no spectral gap between $\lambda_1$ and $\lambda_2$.
For $\gamma = \mathcal{O}(1)$, the curves have deformed in such a way that there is a spectral gap, but $\ell_0$ does not need to change. Thus $k_f=2k_0/N$ can still be a resonant case (Definition \ref{def:resonances}) for large forcing. 
A Ginzburg-Landau equation has been derived using the classical ansatz (\ref{eq:ansatz_cc}) for the case that $k_f = k_0$ and $\gamma = \varepsilon$ in \cite{Ehud_spatialperforcing}. For general systems, Ginzburg-Landau equations are derived for $\gamma=\mathcal{\varepsilon}$ for both $N$ even and odd in \cite{Doelman_Schielen}.

\subsection{Formal derivation of the GL equation in Bloch space}\label{SEC:GL_derivation}
As alluded to, the classical ansatz as used in most works, including \cite{Harten_Validity_GL, SchneiderUecker, Ehud_2008}, does not work due to the size of the forcing. Instead, we use the Bloch transform to derive an ansatz. For the remainder of this section we will assume that Assumption~\ref{ASSUMP:Turing_instability} holds true. Applying the Bloch transform \eqref{eq:Bloch_T} to the Swift-Hohenberg system gives
\begin{align}\label{eq:SH_Bloch_space}
    \partial_t \widetilde{u}(\ell,x,t) = \widetilde{\mathcal{L}}_p\widetilde{u}(\ell,x,t) - \rho(x) (\widetilde{u}(\ell,x,t))^{\star^3} \,,
\end{align}
where the operation $\star$ is the Bloch convolution and $(.)^{\star^m}$ denotes the $(m-1)$-times iteration of the convolution. Using the representation
\begin{align}\label{eq:Bloch_expansion}
    \widetilde{u}(\ell,x,t) = \sum_{n \in \mathbb{N}} \widetilde{A}_n(\ell, t) \psi_n(\ell,x)\,,
\end{align}
the linear part of \eqref{eq:SH_Bloch_space} further simplifies (while the nonlinear terms become more intricate) to
\begin{align}\label{eq:SH_Bloch_space_using_ONB}
    \partial_t \widetilde{A}_n(\ell, t) = \lambda_n(\ell) \widetilde{A}_n(\ell, t) - \left\langle \rho(x) \left(\sum_{j \in \mathbb{N}} \widetilde{A}_j \psi_j \right)^{\star^3}(\ell, t, \cdot), ~\psi_n(\ell, \cdot) \right\rangle \,,
\end{align}
which we abstractly write as
\begin{align}
    \partial_t \widetilde{\mathcal{A}}(\ell, t) =  \widetilde{\mathbf{L}}(\ell)\widetilde{\mathcal{A}}(\ell, t) + \widetilde{\mathbf{N}}(\ell)(\widetilde{\mathcal{A}}(\ell, t))
\end{align}
for $\widetilde{\mathcal{A}} = \left(\widetilde{A}_n \right)_{n \in \mathbb{N}}$. In the following we will craft an ansatz that solves this system ``well enough" to allow a leading order error bound on long time scales. As a measure for this we introduce the residual
\begin{align}\label{eq:res_bloch_diag}
    \widetilde{\mathbf{Res}}\{ \widetilde{\mathcal{A}} \} := - \partial_t \widetilde{\mathcal{A}} + \widetilde{\mathbf{L}}(\ell)\widetilde{\mathcal{A}} + \widetilde{\mathbf{N}}(\ell)(\widetilde{\mathcal{A}}) \, .
\end{align}
The overall structure of the ansatz is
\begin{align}\label{eq:ansatz_bloch_diag}
   \widetilde{\mathcal{A}}_{\mathrm{ansatz}} = 
   \left( 
   \begin{array}{c}
        \widetilde{a}_1 + \varepsilon^2 \widetilde{B}_1  \\
        \left(\varepsilon^2 \widetilde{B}_n\right)_{n > 1}
   \end{array}
   \right)
   \, ,
\end{align}
where $\widetilde{a}_1$ is the ``core ansatz" and $\widetilde{B}_n$ serve to improve the accuracy beyond this ``core ansatz" accounting for nonlinear interactions between different branches of the band structure. The choice of $\varepsilon$-scaling will become clear from the analysis of the nonlinear terms. In order to further refine the ansatz we recall that the Assumptions~\ref{ASSUMP:Turing_instability} demand that $\lambda_n(\ell) < 0$ for $n >1$ and $\lambda_1(\pm \ell_0) = 0, ~ \partial_{\ell}\lambda_1(\pm \ell_0) = 0$, so Taylor expanding $\lambda_1$ around $\pm \ell_0$ results in
\begin{align}
    \lambda_1(\ell) = \varepsilon^2  \left[ r +  \frac12 \partial_{\ell}^2\lambda_1(\pm\ell_0) L_{\pm}^2 \right] + \mathcal{O}(\varepsilon^3)\, ,
\end{align}
with $L_{\pm} = \frac{\ell - (\pm \ell_0)}{\varepsilon}$. Hence, in order to capture the dynamics of linearly unstable wavenumbers beyond onset, we set
\begin{align}\label{eq:Bloch_expansion_ansatz}
    \widetilde{a}_1(\ell,t) &= \widetilde{A}_{1,+}\left( \frac{\ell-\ell_0}{\varepsilon}, ~\varepsilon^2 t\right) + \widetilde{A}_{1,-}\left( \frac{\ell+\ell_0}{\varepsilon}, ~\varepsilon^2  t\right)  \, ,
\end{align}
for which we note that $\widetilde{A}_{1,+}$ is centered around $\ell_0$, while $\widetilde{A}_{1,-}$ is centered around $-\ell_0$ (see Figure \ref{fig:SH_spec_shape_A1}).
\begin{figure}[t]
\centering
\vspace{-1cm}
\scalebox{1}{ 
    \begin{subfigure}[]{0.44\textwidth}
    \includegraphics[width=\textwidth]{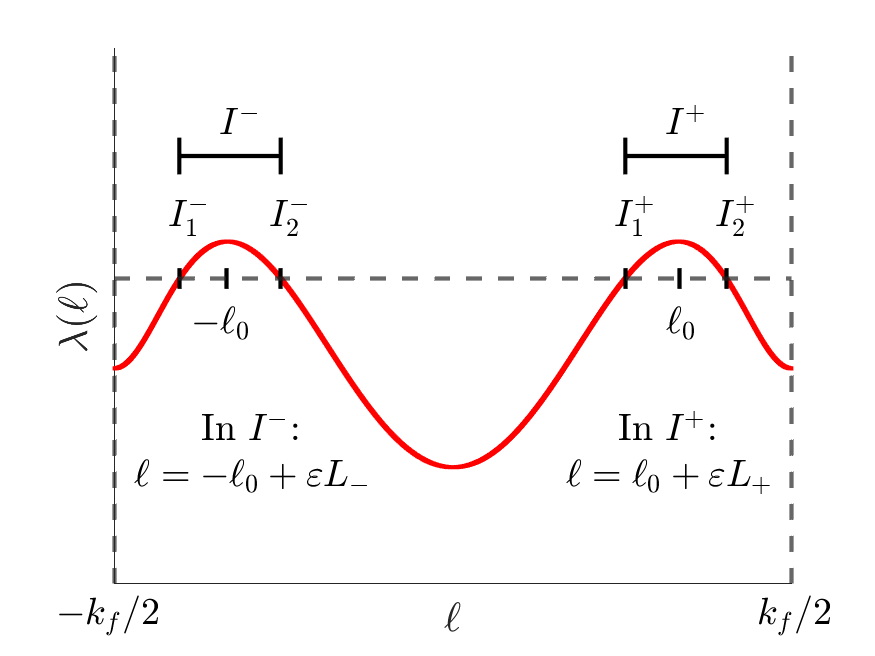}
    \caption{Spectral curve $\lambda_1$.}
   \label{fig:SH_zoom_lambda1}
  \end{subfigure}
    \begin{subfigure}[]{0.44\textwidth}
    \includegraphics[width=\textwidth]{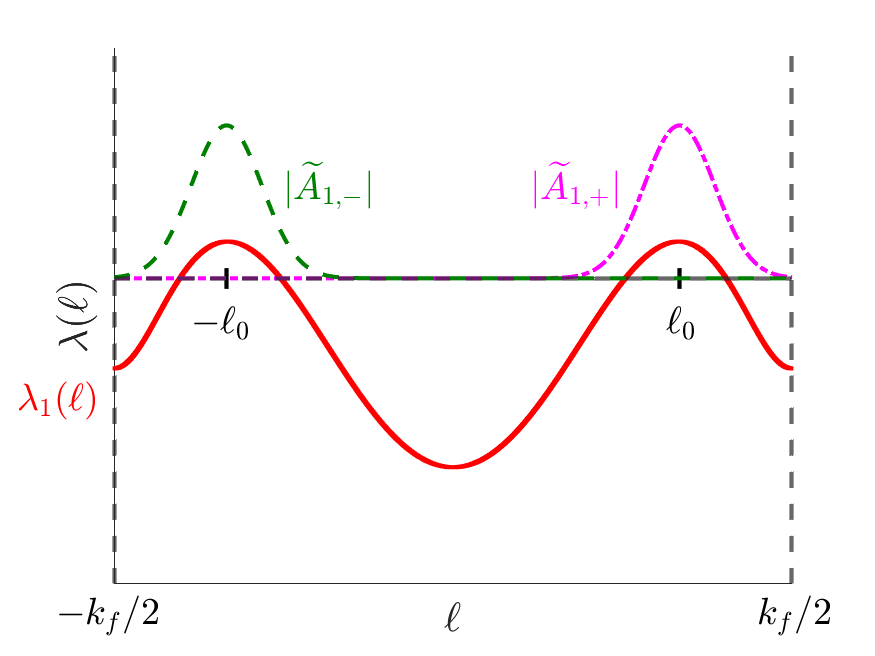}
    \caption{Sketch of the shape of $|\widetilde{A}_{1,+}|$ and $|\widetilde{A}_{1,-}|$ .}
   \label{fig:SH_spec_shape_A1}
  \end{subfigure}
}
\caption{Zoom in of $\lambda_1$, and sketches of $|\widetilde{A}_{1,+}|$ and $|\widetilde{A}_{1,-}|$ with information needed to derive ansatz (\ref{eq:Bloch_expansion_ansatz})}
\end{figure}

After substituting \eqref{eq:ansatz_bloch_diag} into \eqref{eq:res_bloch_diag}, we get
\begin{align*}
&\widetilde{\mathbf{Res}}_1\{\widetilde{\mathcal{A}}_{\mathrm{ansatz}} \} = \left[- \partial_t \widetilde{\mathcal{A}}_{\mathrm{ansatz}} + \widetilde{\mathbf{L}}(\ell)\widetilde{\mathcal{A}}_{\mathrm{ansatz}} + \widetilde{\mathbf{N}}\left(\widetilde{\mathcal{A}}_{\mathrm{ansatz}}\right) \right]_1\\[.2cm]
& = \varepsilon^2 \left[ -\partial_T \widetilde{A}_{1,+} + \left( r +  \frac12 \partial_{\ell}^2\lambda_1(\ell_0) L_{+}^2\right)\widetilde{A}_{1,+} - 3 \left(\int_{0}^{\frac{2 \pi}{k_f}} |\psi_1(\ell_0, x)|^4 \rho(x)~ dx\right) \widetilde{A}_{1,+} \star_{\varepsilon} \widetilde{A}_{1,+} \star_{\varepsilon} \widetilde{A}_{1,-} \right]\\[.2cm]
& + \varepsilon^2 \left[ -\partial_T \widetilde{A}_{1,-} + \left( r +  \frac12 \partial_{\ell}^2\lambda_1(-\ell_0) L_{-}^2\right)\widetilde{A}_{1,-} - 3 \left(\int_{0}^{\frac{2 \pi}{k_f}} |\psi_1(-\ell_0,x)|^4 \rho(x)~ dx \right) \widetilde{A}_{1,+} \star_{\varepsilon} \widetilde{A}_{1,-} \star_{\varepsilon} \widetilde{A}_{1,-} \right]\\[.2cm]
& + \varepsilon^2 \underbrace{\left[ -\left(\int_{0}^{\frac{2 \pi}{k_f}}\psi_1(\ell_0,x)^3 \overline{\psi_1(3\ell_0,x)} \rho(x)~dx\right) \widetilde{A}_{1,+} \star_{\varepsilon} \widetilde{A}_{1,+} \star_{\varepsilon} \widetilde{A}_{1,+} \right]}_{=:\widetilde{B}_{1,3}}\\[.2cm]
& + \varepsilon^2 \underbrace{\left[ -\left(\int_{0}^{\frac{2 \pi}{k_f}} \psi_1(-\ell_0,x)^3 \overline{\psi_1(-3\ell_0,x)} \rho(x)~dx\right) \widetilde{A}_{1,-} \star_{\varepsilon} \widetilde{A}_{1,-} \star_{\varepsilon} \widetilde{A}_{1,-} \right]}_{=:\widetilde{B}_{1,-3}} + \varepsilon^2 \lambda_1(\ell)\widetilde{B}_1 \, 
\end{align*}
for $n=1$, while for $n > 1$ we get
\begin{align*}
&\widetilde{\mathbf{Res}}_n\{\widetilde{\mathcal{A}}_{\mathrm{ansatz}} \} = \left[- \partial_t \widetilde{\mathcal{A}}_{\mathrm{ansatz}} + \widetilde{\mathbf{L}}(\ell)\widetilde{\mathcal{A}}_{\mathrm{ansatz}} + \widetilde{\mathbf{N}}\left(\widetilde{\mathcal{A}}_{\mathrm{ansatz}}\right) \right]_n\\[.2cm]
& = \varepsilon^2 \underbrace{\left[ -3\left(\int_{0}^{\frac{2 \pi}{k_f}} |\psi_1(\ell_0,x)|^2 \psi_1(\ell_0,x) \overline{\psi_n(\ell_0,x)} \rho(x)~dx\right) \widetilde{A}_{1,+} \star_{\varepsilon} \widetilde{A}_{1,+} \star_{\varepsilon} \widetilde{A}_{1,-} \right]}_{=:\widetilde{B}_{n,1}}\\[.2cm]
& + \varepsilon^2 \underbrace{\left[ -3\left(\int_{0}^{\frac{2 \pi}{k_f}} |\psi_1(-\ell_0,x)|^2 \psi_1(-\ell_0,x) \overline{\psi_n(-\ell_0,x)} \rho(x)~dx\right) \widetilde{A}_{1,+} \star_{\varepsilon} \widetilde{A}_{1,-} \star_{\varepsilon} \widetilde{A}_{1,-} \right]}_{=:\widetilde{B}_{n,-1}}\\[.2cm]
& + \varepsilon^2 \underbrace{\left[ -\left(\int_{0}^{\frac{2 \pi}{k_f}} \psi_1(\ell_0, x)^3 \overline{\psi_n(3\ell_0,x)}\rho(x)~ dx\right) \widetilde{A}_{1,+} \star_{\varepsilon} \widetilde{A}_{1,+} \star_{\varepsilon} \widetilde{A}_{1,+} \right]}_{=:\widetilde{B}_{n,3}}\\[.2cm]
& + \varepsilon^2 \underbrace{\left[ -\left(\int_{0}^{\frac{2 \pi}{k_f}} \psi_1(-\ell_0,x)^3\overline{\psi_n(-3\ell_0,x)}\rho(x)~ dx\right) \widetilde{A}_{1,-} \star_{\varepsilon} \widetilde{A}_{1,-} \star_{\varepsilon} \widetilde{A}_{1,-} \right]}_{=:\widetilde{B}_{n,-3}} + \varepsilon^2 \lambda_n(\ell)\widetilde{B}_n \, .
\end{align*}
where we used the notation
\begin{align}\label{eq:epsilon_convonlution}
    \left( \widetilde{A}_{1,+} \star_\varepsilon \widetilde{A}_{1,+} \star_\varepsilon \widetilde{A}_{1,-} \right)(\kappa) 
= \int_{-k_f/(2\varepsilon)}^{k_f/(2\varepsilon)} \int_{-k_f/(2\varepsilon)}^{k_f/(2\varepsilon)} 
\widetilde{A}_{1,+}(\kappa - \kappa_1) \widetilde{A}_{1,+}(\kappa_1 - \kappa_2) \widetilde{A}_{1,-}(\kappa_2) 
\, d\kappa_2 \, d\kappa_1 \,,
\end{align}
which allows to read off that one gains a $\varepsilon^2$ factor in front of the nonlinear terms when changing from $\ell$ to $L_{\pm}$. This justifies the initial structure of the ``core ansatz" \eqref{eq:ansatz_bloch_diag} in retrospect. Following \cite{justification_nonlin_schr} we choose $\widetilde{A}_{1,+} = \chi_{\varepsilon, + \ell_0} \hat{A}$ and $\widetilde{A}_{1,-} = \chi_{\varepsilon, - \ell_0}\hat{\overline{A}}$ with
\begin{align}\label{eq:cut-off}
\chi_{\varepsilon, k}(\ell)
 =
 \left\{
 \begin{array}{rcl}
      1 &,& \varepsilon^{-1}(\ell-k) \in (-\pi/k_f, \pi/k_f] \, ,  \\
      0 &,& \varepsilon^{-1}(\ell-k) \notin (-\pi/k_f, \pi/k_f] \, ,  
 \end{array}
 \right.
\end{align}
and $\widetilde{A}_{1,\pm}$ continued periodically.

\subsubsection{Ginzburg-Landau equation for the non-resonant case}\label{sec:non-resonant_GL}

Assume that $k_f$ is non-resonant, then $\ell_0\neq 3\ell_0\pmod{k_f}$. Demanding that $\widetilde{A}_{1,+}$ solves the Bloch transformed Ginzburg-Landau equation
\begin{align}\label{eq:GL_Bloch}
    \partial_T \widetilde{A}_{1,+} = \left( r +  \frac12 \partial_{\ell}^2\lambda_1(\ell_0) L_{+}^2\right)\widetilde{A}_{1,+} - 3 \left(\int_{0}^{\frac{2 \pi}{k_f}} |\psi_1(\ell_0,x)|^4\rho(x)~ dx\right) \widetilde{A}_{1,+} \star_{\varepsilon} \widetilde{A}_{1,+} \star_{\varepsilon} \widetilde{A}_{1,-} \,
\end{align}
and that $\widetilde{A}_{1,-}$ solves solves the Bloch transformed Ginzburg-Landau equation
\begin{align}\label{eq:GL_Bloch_cc}
    \partial_T \widetilde{A}_{1,-} = \left( r + \frac12 \partial_{\ell}^2\lambda_1(-\ell_0) L_{-}^2\right)\widetilde{A}_{1,-} - 3 \left(\int_{0}^{\frac{2 \pi}{k_f}} |\psi_1(-\ell_0,x)|^4 \rho(x)~dx\right) \widetilde{A}_{1,+} \star_{\varepsilon} \widetilde{A}_{1,-} \star_{\varepsilon} \widetilde{A}_{1,-}
\end{align}
cancels several terms in the residual. However, in order to formally achieve that the residual is of order $\mathcal{O}(\varepsilon^3)$ we further improve the ansatz to
\begin{align}\label{eq:B_n_equations}
\widetilde{B}_1 &= \frac{1}{\lambda_1(3\ell_0)}\widetilde{B}_{1,3} + \frac{1}{\lambda_1(-3\ell_0)}\widetilde{B}_{1,-3}\,, \\[.2cm]
\widetilde{B}_n &= \frac{1}{\lambda_n(\ell_0)}\widetilde{B}_{n,1} + \frac{1}{\lambda_n(-\ell_0)}\widetilde{B}_{n,-1} + \frac{1}{\lambda_n(3\ell_0)}\widetilde{B}_{n,3} + \frac{1}{\lambda_n(-3\ell_0)}\widetilde{B}_{n,-3}\,, \qquad n > 1 \, .
\end{align}
From \eqref{eq:GL_Bloch} we formally find by letting $\varepsilon\rightarrow 0$ the Fourier transformed Ginzburg-Landau equation
\begin{equation}
    \partial_T \hat{A}_{1,+} = \left( r + \frac12 \partial_{\ell}^2\lambda_1(\ell_0) L_{+}^2\right)\hat{A}_{1,+} - 3 \left(\int_{0}^{\frac{2 \pi}{k_f}} |\psi_1(\ell_0,x)|^4\rho(x)~ dx\right) \hat{A}_{1,+} \star_{\varepsilon} \hat{A}_{1,+} \star_{\varepsilon} \hat{A}_{1,-} \, ,
\end{equation}
since $\int_{-k_f/(2\varepsilon)}^{k_f/(2\varepsilon)}\rightarrow \int_\mathbb{R}$. Applying the inverse Fourier transform results in the Ginzburg-Landau equation in $x$-space, equation \eqref{eq:GL_forced_SH}. Combining \eqref{eq:Bloch_expansion}, \eqref{eq:ansatz_bloch_diag}, \eqref{eq:Bloch_expansion_ansatz} and transforming this ansatz back to $x$-space yields
\begin{equation}\label{eq:ansatz_xspace}
\begin{split}
    u_{\rm Ans}(x,t) =&~ u_{\rm GL}(x,t)+ \big( \varepsilon^2 \partial_X A_1 (X,T)\partial_\ell \psi_1(\ell_0,x)e^{i\ell_0x} +\rm c.c. \big)+ \, \mathcal{O}(\varepsilon^3),
\end{split}
\end{equation}
with $u_{\rm GL}$ given in \eqref{eq:u_GL} (conf. \cite{martina_detection}).

\subsubsection{Allen-Cahn derivation}
With some small adaptations, Sections \ref{SEC:derivation} and \ref{sec:justification} can be used to prove Corollary \ref{cor:SH_approx_l00}. In case that $k_f$ and $\mu_*,\gamma_*$ are such that $\ell_0=0$, we adjust the ansatz \eqref{eq:ansatz_bloch_diag} to 
\begin{align}
   \widetilde{\mathcal{A}}_{\mathrm{ansatz}} = 
   \left( 
   \begin{array}{c}
        \widetilde{a}_1  \\
        \left(\varepsilon^2 \widetilde{B}_n\right)_{n > 1}
   \end{array}
   \right)
   \, ,
\end{align}
with
\begin{equation}
    \widetilde{a}_1(\ell,t) = \widetilde{A}_{1}\left(\frac{\ell}{\varepsilon}, \varepsilon^2t\right) \psi_1(\ell,x)
\end{equation}
and demand that $\widetilde{A}_1$ solves the Bloch transformed Allen-Cahn equation
\begin{align}\label{eq:AC_Bloch}
\begin{split}
    \partial_T \widetilde{A}_{1} = &~\left( r +  \frac12 \partial_{\ell}^2\lambda_1(0) L^2\right)\widetilde{A}_{1}
    -\left(\int_{0}^{\frac{2 \pi}{k_f}}\psi_1(0,x)^4 \rho(x)~dx \right) \widetilde{A}_{1} \star_{\varepsilon} \widetilde{A}_{1} \star_{\varepsilon} \widetilde{A}_{1}\,.    
\end{split}
\end{align}
In order to formally achieve that the residual is of order $\mathcal{O}(\varepsilon^3)$ we further improve as before
\begin{equation}\label{eq:Bloch_expansion_ansatz_improved_AC}
\widetilde{B}_n = \frac{1}{\lambda_n(0)}\left(\int_{0}^{\frac{2 \pi}{k_f}}\psi_1(0,x)^3 \psi_n(0,x) \rho(x)dx\right) (\widetilde{A}_{1} \star_{\varepsilon} \widetilde{A}_{1} \star_{\varepsilon} \widetilde{A}_{1}) \,, \qquad n > 1 \, ,
\end{equation}
which is well-defined by Assumption~\ref{ASSUMP:Turing_instability}(g'). From \eqref{eq:AC_Bloch} we formally find the Fourier transformed Allen-Cahn equation
\begin{equation}
\begin{split}
    \partial_T \hat{A}_{1} = &~ \left( r + \frac12 \partial_{\ell}^2\lambda_1(0) L^2\right)\hat{A}_{1} -\left(\int_{0}^{\frac{2 \pi}{k_f}}\psi_1(0,x)^4 \rho(x)dx\right) (\hat{A}_{1} \star_{\varepsilon} \hat{A}_{1} \star_{\varepsilon} \hat{A}_{1})\, .
\end{split}
\end{equation}
This results in the Allen-Cahn equation in $x$-space, equation (\ref{eq:AC_forced_SH}).

\subsubsection{Non-phase invariant Ginzburg-Landau derivation}
With some small adaptations, Sections \ref{SEC:derivation} and \ref{sec:justification} can be used to prove Corollary \ref{cor:SH_approx_res}. Here, as $k_f$ is resonant, we have $\ell_0= 3\ell_0\pmod{k_f}$. We use ansatz \eqref{eq:ansatz_bloch_diag} and demand that $\widetilde{A}_{1,+}$ solves the Bloch transformed non-phase invariant Ginzburg-Landau equation
\begin{align}\label{eq:GL_Bloch_res}
\begin{split}
    \partial_T \widetilde{A}_{1,+} = &~\left( r +  \frac12 \partial_{\ell}^2\lambda_1(\ell_0) L_{+}^2\right)\widetilde{A}_{1,+} - 3 \left(\int_{0}^{\frac{2 \pi}{k_f}} |\psi_1(\ell_0,x)|^4 \rho(x)~dx\right) \widetilde{A}_{1,+} \star_{\varepsilon} \widetilde{A}_{1,+} \star_{\varepsilon} \widetilde{A}_{1,-} \\ & 
    -\left(\int_{0}^{\frac{2 \pi}{k_f}}\psi_1(\ell_0,x)^3 \overline{\psi_1(\ell_0,x)} \rho(x) dx\right) \widetilde{A}_{1,+} \star_{\varepsilon} \widetilde{A}_{1,+} \star_{\varepsilon} \widetilde{A}_{1,+}\, ,    
\end{split}
\end{align}
and that $\widetilde{A}_{1,-}$ solves a similar equation,
\begin{align}\label{eq:GL_Bloch_res_cc}
\begin{split}
        \partial_T \widetilde{A}_{1,-} = &~\left( r + \frac12 \partial_{\ell}^2\lambda_1(-\ell_0) L_{+}^2\right)\widetilde{A}_{1,-} - 3 \left(\int_{0}^{\frac{2 \pi}{k_f}} |\psi_1(-\ell_0,x)|^4 \rho(x)~dx\right) \widetilde{A}_{1,+} \star_{\varepsilon} \widetilde{A}_{1,-} \star_{\varepsilon} \widetilde{A}_{1,-} \\
        & -\left(\int_{0}^{\frac{2 \pi}{k_f}} \psi_1(-\ell_0,x)^3 \overline{\psi_1(-\ell_0,x)} \rho(x)~dx\right) \widetilde{A}_{1,-} \star_{\varepsilon} \widetilde{A}_{1,-} \star_{\varepsilon} \widetilde{A}_{1,-}\, .
\end{split}
\end{align}
In order to formally achieve that the residual is of order $\mathcal{O}(\varepsilon^3)$ we further improve the ansatz as before. Again, Assumption~\ref{ASSUMP:Turing_instability}(g') allows to derive equations for $\widetilde{B}_{n}$ such that the residual is formally of order $\mathcal{O}(\varepsilon^3)$. From \eqref{eq:GL_Bloch_res} we formally find the Fourier transformed non-phase invariant Ginzburg-Landau equation
\begin{equation}
\begin{split}
    \partial_T \hat{A}_{1,+} = &~ \left( r + \frac12 \partial_{\ell}^2\lambda_1(\ell_0) L_{+}^2\right)\hat{A}_{1,+} - 3 \left(\int_{0}^{\frac{2 \pi}{k_f}} |\psi_1(\ell_0,x)|^4 \rho(x)~dx\right) \hat{A}_{1,+} \star_{\varepsilon} \hat{A}_{1,+} \star_{\varepsilon} \hat{A}_{1,-}\\ & -\left(\int_{0}^{\frac{2 \pi}{k_f}}\psi_1(\ell_0,x)^3 \overline{\psi_1(\ell_0,x)} \rho(x)~dx\right) \hat{A}_{1,+} \star_{\varepsilon} \hat{A}_{1,+} \star_{\varepsilon} \hat{A}_{1,+}\, ,
\end{split}
\end{equation}
which in turn results in the non-phase invariant Ginzburg-Landau equation in $x$-space, equation \eqref{eq:GL_forced_SH_res}.

\newpage
\section{Justification of the modulation equations}\label{sec:justification}

To justify the Ginzburg-Landau approximation, it is required to prove error estimate (\ref{eq:thm_error_est}). To that end, we use the function space $H_{ul}^\vartheta$ as introduced in Section \ref{subsec:function_spaces}. In this function space, we give a rigorous estimate of the magnitude of the residual. We show that the semigroup generated by the linear operator is strongly continuous, such that the variation of constant formula can be applied. The remaining terms in the error equation have to be bounded, such that applying Grönwalls theorem gives the desired estimate. Throughout this section many different constants that are independent of $\varepsilon$ are denoted by C.\\

To give an estimate on how close the approximation $u_{\rm Ans}$ is to the real solution, we have to estimate the error
\begin{equation}\label{eq:error_S_def}
    \varepsilon^\beta S(x,t) = u(x,t) - u_{\rm Ans}(x,t).
\end{equation}
The error $S$ satisfies the PDE
\begin{equation}\label{eq:error_eq}
\begin{split}
    \partial_t S = & -(k_0^2+\partial_x^2)^2 S + (p_*+\varepsilon^2r) S -\varepsilon^{2\beta}\rho~ S^3\\ & -3 \varepsilon^\beta \rho S^2 u_{\rm Ans} - 3\rho S u_{\rm Ans}^2-\varepsilon^{-\beta} \text{Res}(u_{\rm Ans}).
\end{split}
\end{equation}
In order to prove Theorem \ref{thm:SH_approx}, we want to bound all the terms in this equation. Introduce
\begin{equation}\label{eq:def_F}
    F(S) : = \varepsilon^2 r S -\varepsilon^{2\beta} \rho S^3-3\varepsilon^\beta \rho S^2 u_{\rm Ans} - 3\rho S u_{\rm Ans}^2-\varepsilon^{-\beta} ~\text{Res}(u_{\rm Ans}).
\end{equation}
Then the error equation is given by
\begin{equation}\label{eq:error_eq_S}
    \partial_t S = \mathcal{L}_{p_*} S +F(S).
\end{equation}
By standard arguments it can be shown that $\mathcal{L}_{p_*}$ generates a strongly continuous semigroup of uniformly bounded linear operators. We can split $\mathcal{L}_{p_*}$ into $\mathcal{L}_{p_*} = \mathcal{L}_{\rm SH} + B$ with $\mathcal{L}_{\rm SH} = -(k_0^2+\partial_x^2)^2$ and $B=p_*(x)$. Then we know that $\mathcal{L}_{\rm SH}$ generates a strongly continuous semigroup. As $B$ is a bounded perturbation, it follows from \cite[Theorem 1.3]{Engel_Nagel} that $\mathcal{L}_{p_*}$ generates a strongly continuous semigroup.

In order to give a semigroup bound we consider the semigroup in Bloch space, generated by $\widetilde{\mathcal{L}}_{p_*}$, and examine its spectrum. Assumption~\ref{ASSUMP:Turing_instability}(f) allows a separation of the diffusive mode from the exponentially damped modes. Taking the Bloch transform of equation (\ref{eq:error_eq_S}) yields
\begin{equation}\label{eq:VOC_Bloch}
    \partial_t \widetilde{S} = \widetilde{\mathcal{L}}_{p_*} \widetilde{S} +\widetilde{F}(\widetilde{S}).
\end{equation}
Similarly to \cite{Schneider_Maier_CB}, as $\widetilde{\mathcal{L}}_{p_*}$ is self-adjoint, we define the orthogonal projection on the diffusive and the exponentially damped modes by
\begin{equation}
    \widetilde{P}_c(\ell)\widetilde{S}(\ell) = \left\langle \psi_1(\ell), \widetilde{S}(\ell)\right\rangle \psi_1(\ell), 
\end{equation}
\begin{equation}
    \widetilde{P}_s(\ell)\widetilde{S}(\ell) = \widetilde{S}(\ell)- \widetilde{P}_c(\ell)\widetilde{S}(\ell). 
\end{equation}
This separates (\ref{eq:VOC_Bloch}) into
\begin{equation}
    \partial_t \widetilde{S}_q (\ell,x,t) = \widetilde{\mathcal{L}}_{p_*,q}(\ell)\widetilde{S}_q(\ell,x,t) + \widetilde{P}_q(\ell)\widetilde{F}(\widetilde{S})(\ell,x,t), \quad \text{ for } q=c,s,
\end{equation}

where $\widetilde{\mathcal{L}}_{p_*,q}(\ell) = \widetilde{\mathcal{L}}_{p_*}(\ell) \widetilde{P}_q(\ell)$. As the Bloch transform is an isomorphism between $H_{ul}^\vartheta$ and \\ $L^2_{ul}([-\pi/\alpha, \pi/\alpha), H^\vartheta_{ul}((0,\alpha], \mathbb{C}))$, we consider the space $L^2_{ul}([-\pi/\alpha, \pi/\alpha), H^\vartheta_{ul}((0,\alpha], \mathbb{C}))$.  Note that the operator $\widetilde{\mathcal{L}}_{p_*,s}(\ell)$ is sectorial and has spectrum in the left half plane, bounded away from the imaginary axis. Hence the following estimate follows \cite[Theorem 1.5.3]{Henry}.

\begin{lemma}\label{lem:estimate_Ls}\emph{\textbf{(Estimate of the sectorial part)}}
    For the analytic semigroup generated by $\widetilde{\mathcal{L}}_{p_*,s}$ we have
    \begin{equation}
        \|e^{\widetilde{\mathcal{L}}_{p_*,s} t}\|_{L_{ul}^2} \leq Ce^{\sigma t/2}\leq C,
    \end{equation}
    for some $\sigma>0$, $C>0$ and all $t\geq 0$.
\end{lemma}

For the diffusive mode we find the following estimate.
\begin{lemma}\label{lem:estimate_Lc}\emph{\textbf{(Estimate of them   diffusive part)}}
    For the semigroup generated by $\widetilde{\mathcal{L}}_{p_*,c}$ we have
    \begin{equation}
    \sup_{t\in [0,~T_0/\varepsilon^2]} \|e^{\widetilde{\mathcal{L}}_{p_*,c} t } \|_{L_{ul}^\vartheta} \leq C,
    \end{equation}
    for some $C>0$ and all $t\geq 0$.
\end{lemma}

\noindent\textbf{Proof.} Note that 
\begin{equation}
    \|e^{\widetilde{\mathcal{L}}_{p_*,c} t } \widetilde{S}_c\|_{H_{ul}^\vartheta} = \|e^{\lambda_1(\ell) t }\widetilde{S}_c \|_{H_{ul}^\vartheta} \leq \|e^{\lambda_1(\ell) t}\|_{L^\infty} \|\widetilde{S}_c \|_{H_{ul}^\vartheta}.
\end{equation}

Furthermore, note that $\lambda_1(\ell)<-\sigma_0$ for some $\sigma_0>0$ and $\ell\notin I^-\cup I^+$. For $\ell \in I^- \cup I^+$ we can use $\lambda_1(\ell)\leq -C(\ell\pm \ell_0)^2 = - C(\varepsilon L_\pm)^2$ for some $C>0$ (Assumption~\ref{ASSUMP:Turing_instability}(e)). 
Since $\lambda_1(\ell)\approx -C(\varepsilon L)^2$ for small $L$, the term $e^{\lambda_1(\ell)t}$ behaves as $e^{-L^2T}$ with $T=\varepsilon^2t$ so its $C_b^2$-norm is uniformly bounded for $t$ on an $\mathcal{O}(1/\varepsilon^2)$ timescale. Thus for $t\geq 0$ we find
\begin{equation}
     \sup_{t\in [0,~T_0/\varepsilon^2]}\|e^{\widetilde{\mathcal{L}}_{p_*,c} t }\widetilde{S}_c \|_{H_{ul}^\vartheta} \leq \sup_{t\in [0,~T_0/\varepsilon^2]} \| e^{\lambda_1(\ell) t}\|_{L^\infty} \|\widetilde{S}_c \|_{H_{ul}^\vartheta}\leq C  \|\widetilde{S}_c \|_{H_{ul}^\vartheta}.
\end{equation}
\qed

\noindent Together, this yields the following lemma.

\begin{lemma}\label{lem:semigroup_estimate} \emph{\textbf{(Semigroup estimate)}}
    For every $\vartheta>1/2$ there exists a $C>0$ such that for all $\varepsilon\in (0,1]$ the semigroup $(e^{t\mathcal{L}_{p_*}})_{t\geq 0}$ generated by $\mathcal{L}_{p_*}$ satisfies
    \begin{equation}
        \sup_{t\in [0,~T_0/\varepsilon^2]} \|e^{t\mathcal{L}_{p_*}}\|_{H_{ul}^\vartheta\rightarrow H_{ul}^\vartheta}\leq C.
    \end{equation}
\end{lemma}

\noindent\textbf{Proof.}  
Let $Z(t): H^\vartheta_{ul}\rightarrow H^\vartheta_{ul}$ be the semigroup generated by $\mathcal{L}_{p_*}$. This semigroup is defined by $Z(t)= \mathcal{B}\widetilde{Z}(t)\mathcal{B}^{-1}$, where $\widetilde{Z}(t) = e^{\widetilde{\mathcal{L}}_{p_ *}t}$ for $t\geq 0$.

Applying Lemmas \ref{lem:estimate_Ls} and \ref{lem:estimate_Lc} yields the estimate
\begin{equation}
\begin{split}
     \sup_{t\in [0,~T_0/\varepsilon^2]}  \|Z(t)u\|_{H^\vartheta_{ul}} \leq & \sup_{t\in [0,~T_0/\varepsilon^2]} C \|\widetilde{Z}(t)\widetilde{u}\|_{L^2_{ul}}\\  \leq & \sup_{t\in [0,~T_0/\varepsilon^2]} C\|\widetilde{u}\|_{L_{ul}^2} \leq  \sup_{t\in [0,~T_0/\varepsilon^2]} C \|u\|_{H_{ul}^\vartheta}.  
\end{split}
\end{equation}
\qed

Hence the variations of constant formula is applicable \cite{Sell2002} and we find
\begin{equation}\label{eq:error_var_const}
    S(t) = e^{t\mathcal{L}_{p_*}} S(0) + \int_0^t e^{(t-\tau)\mathcal{L}_{p_*}} F(S)(\tau)~d\tau.
\end{equation}

\noindent To bound $S(t)$, we bound the different terms in $F$. For that we in particular need an estimate on the residual. We will formulate our results in the following for $u_{\rm Ans}$ built on $u_{\rm GL}$, but the analogous holds true for $u_{\rm AC}$.

\begin{lemma}\label{lem:res_est}\emph{\textbf{(Residual estimate)}}
Let $\vartheta\geq 1$ and let $A\in C([0,T_0], H_{ul}^{\vartheta_A})$ with $\vartheta_A = 3+\vartheta $ be a solution of the Ginzburg-Landau equation (\ref{eq:GL_forced_SH}) for $k_f$ non-resonant, or of (\ref{eq:GL_forced_SH_res}) for $k_f$ resonant. Then for all $\varepsilon_0\in (0,1]$, there exists a $C>0$ such that for all $\varepsilon\in (0,\varepsilon_0)$ we have 
    \begin{equation}
        \sup_{t\in[0,~T_0/\varepsilon^2]}\|{\rm Res}(u_{\rm Ans}(\cdot, t))\|_{H_{ul}^\vartheta}\leq C \varepsilon^4,   
    \end{equation}
    and 
    \begin{equation}
        \sup_{t\in[0,~T_0/\varepsilon^2]}\| u_{\rm Ans}(\cdot, t) - u_{\rm GL}(\cdot, t)\|_{H_{ul}^\vartheta}\leq C \varepsilon^{2}.
    \end{equation}
\end{lemma}

\begin{proof}
Following \cite{justification_nonlin_schr}, Lemma 3.3 and Lemma 4.3 allows to bound the residual \eqref{eq:res_bloch_diag}. Furthermore, tracking the estimates for the various transformations using the reasoning in \cite[Section 5]{martina_detection} allows to port the estimate to the residual of the original equation \eqref{eq:error_eq}. In particular, note that the Bloch transform is an isomorphism between $H_{ul}^\vartheta$ and $L^2_{ul}([-k_f/2, k_f/2), H^\vartheta_{ul}((0,2\pi/k_f], \mathbb{C}))$ \cite{Scarpellini, Simon_Reed2} and that the diagonalization operator formalizing the expansion w.r.t. the ONB $(\psi_n)_{n \in \mathbb{N}}$ is bounded \cite[Lemma 4.3]{justification_nonlin_schr}. Finally, the estimate between $u_{\rm Ans}$ and $u_{\rm GL}$ can be deduced from the expansion of the eigenfunction $\psi_1$, as given in \eqref{eq:ansatz_xspace}.
\end{proof}

The remaining terms in (\ref{eq:def_F}) can be estimated as follows.

\begin{lemma}\label{lem:estimate_F}\emph{\textbf{(Estimate of $\bm{F}$)}}
    Let $\vartheta\geq 1$, $\vartheta_A=\vartheta+3$ and $A\in C([0,T_0], H_{ul}^{\vartheta_A})$ be a solution of the Ginzburg-Landau equation (\ref{eq:GL_forced_SH}) for $k_f$ non-resonant, or of (\ref{eq:GL_forced_SH_res}) for $k_f$ resonant. Then there is a $C>0$ such that for all $\varepsilon\in(0,1]$ we have
    \begin{equation}
        \|F\|_{H_{ul}^\vartheta}\leq C(\varepsilon^2 \|S\|_{H_{ul}^\vartheta} +\varepsilon^{2\beta}\|S\|_{H_{ul}^\vartheta}^3+  \varepsilon^{\beta +1} \| S\|_{H_{ul}^\vartheta}^2 + \varepsilon^2).
    \end{equation}
\end{lemma}

\noindent\textbf{Proof.} First note that we can bound
\begin{equation}
    \|\rho S\|_{H_{ul}^\vartheta}\leq C\|\rho\|_{C_b^\vartheta} \|S\|_{H_{ul}^\vartheta}\leq C \|S\|_{H_{ul}^\vartheta},
\end{equation}
since $\rho$ is continuous and periodic. The different terms in $F$ can be bounded by
\begin{equation}
    \begin{split}
        \|\rho S^3\|_{H_{ul}^\vartheta} \leq & ~C\| S \|_{H_{ul}^\vartheta}^3,\\
        \|\rho ~u_{\rm Ans} S^2\|_{H_{ul}^\vartheta} \leq & ~C\|u_{\rm Ans}\|_{C_{b}^\vartheta} \|S\|_{H_{ul}^\vartheta}^2\\
        \|\rho ~u_{\rm Ans}^2 S\|_{H_{ul}^\vartheta}\leq & ~C\|u_{\rm Ans}\|_{C_{b}^\vartheta}^2\|S\|_{H_{ul}^\vartheta},\\
        \|\varepsilon^{-\beta} \text{Res}(u_{\rm Ans})\|_{H_{ul}^\vartheta}\leq & ~C\varepsilon^{4-\beta}\leq C\varepsilon^2, \text{ for $\beta$ chosen such that $\beta\leq 2$}.
    \end{split}
\end{equation}
The last estimate follows from Lemma \ref{lem:res_est}. We moreover have
\begin{equation}
    \|u_{\rm Ans}\|_{C_{b}^\vartheta}=\|x\mapsto \varepsilon(A_1(\varepsilon x)\psi_1(\ell_0,x) e^{i\ell_0x}+\rm{c.c.}+\rm{h.o.t.})\|_{C_{b}^\vartheta}\leq C\varepsilon\|A_1\|_{C_{b}^\vartheta}.
\end{equation}
From Sobolev's embedding theorem we have
\begin{equation}
    \|A\|_{C_{b}^\vartheta}\leq C\|A \|_{H_{ul}^{\vartheta+1}},
\end{equation}
which concludes the proof.\qed\\

We now have all the ingredients to prove Theorem \ref{thm:SH_approx}.\\

\noindent\textbf{Proof of Theorem \ref{thm:SH_approx}.} We will bound the mild solutions of the equations for the error $S$ in (\ref{eq:error_var_const}). Local existence and uniqueness of solutions in $H_{ul}^\vartheta$ for $\vartheta>\frac{1}{2}$ now follows from Lemma \ref{lem:semigroup_estimate} together with \cite[Theorem 5.2.22]{SchneiderUecker}. We have to ensure the norm of the solution stays finite for $t\in [0,~T_0/\varepsilon^2]$. Using the estimates from Lemma \ref{lem:semigroup_estimate} and \ref{lem:estimate_F} together with equation (\ref{eq:error_var_const}) we obtain
\begin{equation}
\begin{split}
    \|S(t)\|_{H_{ul}^\vartheta} & \leq \|e^{t\mathcal{L}_{p_*}} S(0)\|_{H_{ul}^\vartheta} + \\ & C \varepsilon^2 \int_0^t \|e^{(t-\tau)\mathcal{L}_{p_*}}\|_{H_{ul}^\vartheta} \left(\|S(\tau)\|_{H_{ul}^\vartheta} +\varepsilon^{2\beta-2} \|S(\tau)\|_{H_{ul}^\vartheta}^3 +\varepsilon^{\beta-1} \|S(\tau)\|^2_{H_{ul}^\vartheta}+1
    \right)~d\tau \\
    \leq & C \varepsilon^2 \int_0^t \left(\|S(\tau)\|_{H_{ul}^\vartheta} +\varepsilon^{2\beta-2} \|S(\tau)\|_{H_{ul}^\vartheta}^3 +\varepsilon^{\beta-1} \|S(\tau)\|^2_{H_{ul}^\vartheta}+1
    \right)~d\tau.
\end{split}
\end{equation}
Imposing
\begin{equation}\label{eq:proof_condition1}
    \varepsilon^{2\beta-2} \|S(\tau)\|_{H_{ul}^\vartheta}^3 +\varepsilon^{\beta-1} \|S(\tau)\|^2_{H_{ul}^\vartheta}\leq 1,
\end{equation}
yields
\begin{equation}
    \|S(t)\|_{H_{ul}^\vartheta} \leq C\varepsilon^2\int_0^t (\|S(\tau)\|_{H_{ul}^\vartheta}+2)~d\tau\leq 2CT_0+C\varepsilon^2\int_0^t\|S(\tau)\|_{H_{ul}^\vartheta}~d\tau.
\end{equation}

\noindent Applying Gronwall's inequality for integrals gives
\begin{equation}
    \|S(t)\|_{H_{ul}^\vartheta}\leq 2CT_0e^{C\varepsilon^2t}\leq 2CT_0e^{CT_0}=:M 
\end{equation}
for all $t\in[0,~T_0/\varepsilon^2]$. In order to fulfill requirement (\ref{eq:proof_condition1}), we choose $\varepsilon_0>0$ such that $\varepsilon_0^{2\beta-2}M^3+\varepsilon_0^{\beta-1}M^2\leq 1$. As we want $\beta$ to be as large as possible, but $\beta\leq 2$ (Lemma \ref{lem:estimate_F}), we take $\beta = 2$. This gives the estimate
\begin{equation}\label{eq:thm_eq_1}
        \sup_{t\in [0,~T_0/\varepsilon^2]} \|u(t)-u_{\rm Ans}(t)\|_{H_{ul}^\vartheta} \leq C\varepsilon^{2}.
    \end{equation}

\noindent The triangle inequality now gives
\begin{equation}
\begin{split}
    \sup_{t\in [0,~T_0/\varepsilon^2]}& \|u(t)-u_{\rm GL}(t)\|_{H_{ul}^\vartheta} \\ &\leq \sup_{t\in[0,~T_0/\varepsilon^2]} \|u(t)-u_{\rm Ans}(t)\|_{H_{ul}^\vartheta} +\sup_{t\in[0,~T_0/\varepsilon^2]} \|u_{\rm Ans}(t)-u_{\rm GL}(t)\|_{H_{ul}^\vartheta}\\ & \leq C\varepsilon^2 + C\varepsilon^{2} \leq C \varepsilon^{2}
\end{split}
\end{equation}
This proves the approximation theorem (Theorem \ref{thm:SH_approx}) with function space $Y=H_{ul}^\vartheta$ and $A_1\in C([0,T_0], H_{ul}^{\vartheta_A})$, $\vartheta_A=\vartheta+3$ and $\vartheta\geq 1$. \qed

\section{Related work}\label{sec:related_work}
There is a wealth of literature on the derivation of modulation equations for many different systems. While formal derivations have a long tradition in the physics and beyond, a rigorous treatment adding validity proofs with error estimates started in the early 1990s (e.g. \cite{SH_GL_val_cubic}). 

\medskip

We refrain from citing all different systems for which the Ginzburg-Landau equation arises (see, e.g. \cite{SchneiderUecker} for a thorough treatment of the subject). Staying in the context of Swift-Hohenberg equations, we direct attention to the following: In addition to modifying the linear part of the equation by applying a forcing term, non-local effects can be incorporated into other parts of the equation, leading to different alterations in the derived amplitude equation. An amplitude equation for the Swift-Hohenberg system with fractional Laplacian has been derived and justified in \cite{Kuehn_SH_diff}, while a Swift-Hohenberg system with non-local nonlinearity is considered in \cite{Kuehn_SH_nonlocal_nonlin}. In contrast to these works, which utilise the classical ansatz (\ref{eq:ansatz_cc}), we showed that adding a large forcing to the linearity results in a modification of the ansatz.  

\medskip

In addition to the Ginzburg-Landau equation, other types of modulation equations can be derived for e.g. dispersive systems with periodic coefficients. The Nonlinear Schr\"odinger Equation has been derived and justified for a $\phi^4$-model with spatially periodic coefficients in \cite{justification_nonlin_schr}. Some of the ideas from that work are applied here.
Heterogeneous media can also influence pulse propagation, depending on the size of the period of the heterogeneity, as shown in \cite{Nishiura_etal}.  Modulation equations for spatially periodic wavetrains in nonlinear partial differential equations are derived and analysed in \cite{DSSS_Dynamics_ModulatedWT}.

\medskip

This work is partially motivated by the study of dryland ecosystems, where vegetation patterns arise due to scarcity of water. In many dryland ecosystem models, such as the Klausmeier model \cite{Klausmeier1999}, it is assumed that the underlying terrain is flat. However, aerial images show that stripe patterns are prevalent on slopes, while hexagonal pattern structures are most commonly found on flat ground. Furthermore, numerical studies show differences in the movement of vegetation patches on slopes compared to flat terrain. These observations underpin the importance of topography to the possible vegetation patterns. Incorporating topography into a vegetation model results in a system with spatially varying coefficients. Keeping this application in mind, one can think of the periodic forcing in \eqref{EQ:SH_periodic_coefficients} as a periodic topography. 

\subsection{Modulation equations for small forcing}
As alluded to, the Ginzburg-Landau equation and variants thereof have already been derived for systems with spatially periodic forcing. However, to the best of our knowledge, these derivations have been limited to systems with small forcing strength. Furthermore, the validity of the derived Ginzburg-Landau equation in systems with spatial forcing has not been established for many of such systems. Ginzburg-Landau equations are derived (but not validated) for \eqref{EQ:SH_periodic_coefficients} with $p(x)=\varepsilon^2r+\gamma \cos(k_fx)$ and $\gamma$ of order $\mathcal{O}(\varepsilon^k)$ for $k=1,2$ in \cite{Ehud_2008, Ehud_2012, Ehud_spatialperforcing, Ehud_2014, Ehud_2015_1, uecker2001}. They have also been derived for a class of reaction-diffusion equations with quadratic nonlinearity and small varying terrain in \cite{Eckhaus_Kuske, Doelman_Schielen}.

\medskip

In \cite{Ehud_spatialperforcing}, the Swift-Hohenberg equation \eqref{EQ:SH_periodic_coefficients} with $\gamma = \sum_{i=1}^\infty |\varepsilon|^i \gamma_i$, $\gamma_i\sim\mathcal{O}(1)$ is considered. Without studying the linear problem, the classical multiple scaling ansatz (\ref{eq:ansatz_cc}) is used. For the non-resonant case, this leads to the Ginzburg-Landau equation
\begin{equation}\label{eq:GL_Ehud_nonres}
    \partial_T A = \left(r+\left(\frac{\gamma_1}{2}\right)^2(d_++d_-)\right)A + 4k_0^2 \partial_X^2 A - 3|A|^2A, \quad \text{with }d_\pm = \frac{1}{k_f^2(k_f\pm 2k_0^2)}. 
\end{equation}
Since our method does not exclude $\gamma$ small, but only uses assumptions on the spectrum, both methods should give similar results for $k_f$ non-resonant and $\gamma\ll 1$. In fact, using a Bloch wave ansatz and setting $\psi_1= 1 + \varepsilon \widetilde{\psi_1}$ should give the analogous result.

\medskip

In \cite{Ehud_spatialperforcing}, the main focus lies on resonances $k_f\approx nk_0$. A Ginzburg-Landau equation has been derived using the classical ansatz (\ref{eq:ansatz_cc}) for $k_f\approx k_0$ and $\gamma\sim \mathcal{O}(\varepsilon)$ and $k_f\approx 2k_0$ and $\gamma\sim\mathcal{O}(\varepsilon^2)$. In addition, the 2D case is covered for weak forcing $\gamma\ll1$. 
This indicates that it is possible to do the multidimensional case for larger forcing as well (see, e.g. \cite{Hoyle2006Pattern}, \cite{Gaff2023}).
\medskip

In \cite{Eckhaus_Kuske}, a class of equations that describe the evolution of patterns in problems with a slightly varying geometry is studied. Compared to our paradigm model, this relates to \eqref{EQ:SH_periodic_coefficients} with $p(x) = \varepsilon^2 r +\varepsilon^2 \gamma_2 f(k_fx)$ with $\gamma_2\sim\mathcal{O}(1)$, $k_f \sim o(1)$ and $f$ a function that admits a Fourier expansion $f(k_fx) = \sum_{n=-\infty}^{\infty} F_n e^{ink_f x}$. Similar as in \cite{Ehud_spatialperforcing}, the classical multiple scaling ansatz is used without analyzing the linear problem. The latter can be justified since the forcing only has influence at order $\mathcal{O}(\varepsilon^3)$ in the Ginzburg-Landau derivation. As $k_f=o(1)$, this yields the Ginzburg-Landau equation with spatially varying coefficients
\begin{equation}\label{eq:GL_Kuske}
    \partial_T A= \left(r +\gamma_2\sum_{n=-\infty}^\infty F_n e^{ink_f X/\varepsilon} \right)A + 4k_0^2 \partial_X^2A-3|A|^2A,
\end{equation}
which exists if $\sum_{n=-\infty}^\infty F_n e^{ink_f X/\varepsilon}$ converges in the appropriate function space. It is important to note that (\ref{eq:GL_Kuske}) has an explicit $\varepsilon$-dependency, which can be compensated by rewriting $k_f = \varepsilon \sigma_0$. .

\medskip

In contrast to the works mentioned above, linear analysis is carried out and used to derive an amplitude equation in \cite{Doelman_Schielen}. In the non-resonant cases, with forcing $\gamma\ll1$, this yields an approximation
\begin{equation}
    u_{\rm Ans}(x,t) = \varepsilon A(X,T) e^{i\ell_0x}\sum a_n e^{ink_fx} +c.c. + h.o.t.,  
\end{equation}
where the coefficients $a_n$ are yet unknown and where $A$ satisfies the Ginzburg-Landau equation
\begin{equation}\label{eq:GL_Schielen}
    \partial_TA = r\partial_R\lambda_1(\ell_0) A -\frac{1}{2}\partial_\ell^2 \lambda_1(\ell_0) \partial_X^2 A -3(1+\gamma^2\beta_2)|A|^2A,
\end{equation}
with $\beta_2$ a complex number. As the Bloch wave ansatz does not exclude small wavenumbers, both approaches should give a similar result for small forcing.

\medskip

\medskip
   
In \cite{Ehud_spatialperforcing, Eckhaus_Kuske}, no justification of the Ginzburg-Landau approximation has been done. To show validity of the approximations for the systems with forcings of order $\gamma\sim\mathcal{O}(\varepsilon^2)$, for a great part the proof of the constant coefficient case can be followed. We introduce an error $S$ as in (\ref{eq:error_S_def}), which satisfies the PDE (\ref{eq:error_eq}) with $\mu_*=0$ and $\gamma_*=\varepsilon^2\gamma_2$. In case of $p(x)=\varepsilon^2 r + \varepsilon^2\gamma_2 f(k_fx)$ as above, the $\cos(k_fx)$ term is replaced by $f(k_fx)$. Upon taking $\mathcal{L}:= -(k_0^2+\partial_x^2)^2$, the semigroup generated by $\mathcal{L}$ is strongly continuous and can be bounded \cite{SchneiderUecker}. The residual can be decreased to order $\mathcal{O}(\varepsilon^4)$ using standard techniques. Lastly, the forcing term $\varepsilon^2\gamma_2f(k_f x)$ is included into the equation for $F(S)$, which satisfies Lemma \ref{lem:estimate_F}. Following the rest of the proof in Section \ref{sec:justification} leads to validity of the approximations in \cite{Ehud_spatialperforcing, Eckhaus_Kuske}.

The validity of the approximation for forcing $\gamma\sim\mathcal{O}(\varepsilon)$, done in \cite{Ehud_spatialperforcing}, requires more work. We again introduce an error $S$, defined in (\ref{eq:error_S_def}), satisfying (\ref{eq:error_eq}) with $\mu_*=0$ and $\gamma_*=\varepsilon\gamma_1$. Due to the size of the forcing, $\varepsilon\gamma_1\cos(k_fx)S$ cannot be part of $F(S)$ anymore. Instead, we can consider $\mathcal{L}_{p_*} := -(k_0^2+\partial_x^2)^2 + \varepsilon\gamma_1\cos(k_fx)$. Similarly as before, we now have to show that the semigroup generated by $\mathcal{L}_{p_*}$ is strongly continuous and can be bounded for $t\in[0, T_0/\varepsilon]$. To that end, we need to understand the spectrum, hence proving validity without doing spectral analysis is not possible anymore. For $k_f$ non-resonant, Lemma \ref{lem:semigroup_estimate} still holds true. However, for $k_f$ resonant, due to the absence of the spectral gap between the first two spectral curves, the splitting argument in the proof of Lemma \ref{lem:semigroup_estimate} needs to be adjusted.

\subsection{Connection to stability analysis of periodic solutions of the Swift-Hohenberg equation}

The Swift-Hohenberg equation with constant coefficients possesses stationary periodic solutions that bifurcate for $r=0$ from the background state $u=0$ \cite{Schneider_Diffusive_stab, eckmann1997geometric}. For $r=1$, these solutions are of the form
\begin{equation}
    u_{\rm per}(\omega, \varphi, \varepsilon)[x]=\varepsilon(\sqrt{1-4\omega^2}/\sqrt{3})e^{i(k_0+\varepsilon\omega)x}e^{i\varphi}+c.c.+\mathcal{O}(\varepsilon^2),
\end{equation}
with $\omega,\phi\in\mathbb{R}$.
Linearising around $u_{\rm per}$ results in the linear operator $\mathcal{L}_{\rm per}$, given by
\begin{equation}
    \mathcal{L}_{\rm per} = -(k_0^2+\partial_x^2)^2 +\varepsilon^2 -3u_{\rm per}^2.
\end{equation}
For $\varphi=0$, we find
\begin{equation}
    \mathcal{L}_{\rm per} = -(k_0^2+\partial_x^2)^2 +\varepsilon^2\left[ 1-2(1-4\omega^2) -2(1-4\omega^2)\cos(2k_0+2\varepsilon\omega)\right]+\mathcal{O}(\varepsilon^3)\, .
\end{equation}
Note that $\mathcal{L}_{\rm per} = \mathcal{L}_p$ for 
\begin{equation}
    p(x)= \varepsilon^2\left(1-2(1-4\omega^2)\right) -2\varepsilon^2(1-4\omega^2)\cos(2k_0+2\varepsilon\omega)+\mathcal{O}(\varepsilon^3),
\end{equation}
meaning that the linearisation around $u_{\rm per}$ gives the same spectrum as the forced Swift-Hohenberg system where both $\mu$ and $\gamma$ are of order $\mathcal{O}(\varepsilon^2)$ and $k_f\approx 2k_0$ is resonant. Hence the spectrum is similar to the spectrum in Figure \ref{fig:SH_spec_res_cases} (bottom left). In this case, it is known that the derivative of $u_{\rm per}$ is the eigenfunction corresponding to eigenvalue $0$. Therefore, it is possible to compute $\psi_1$ analytically for this specific case.

\newpage

\section{Discussion and Outlook}\label{sec:discussion}

We have proven that, even in the presence of large forcing, solutions of a periodically forced Swift-Hohenberg system can be approximated on long finite timescale by solutions of a Ginzburg-Landau equation. Compared to the constant coefficient or small forcing case, the multiple scaling ansatz must be adjusted to a Bloch wave structure. At leading order, this approximation depends on the first Bloch eigenfunction, and the coefficients of the Ginzburg-Landau equation therein depend on the shape of the first spectral curve. We anticipate many different extensions of the present work.

\subsection{Resonances without spectral gap}\label{sec:resonant_extension}
Throughout this work, we assume that the spectrum satisfies Assumption~\ref{ASSUMP:Turing_instability}. As explained Section~\ref{sec:resonances}, further research is required to derive and justify the Ginzburg-Landau equations in case there is no spectral gap between the first two spectral curves (Assumption~\ref{ASSUMP:Turing_instability}(f)), which often occurs for resonant wavenumbers. 
Other potential extensions of this work are considering $k_f\gg 1$, or using $\widetilde{p}$ to create the bifurcation (see Remark \ref{rem:p_tilde}). While the latter primarily increases the complexity of the derivation, we do not expect it to introduce new phenomena. 

\medskip

\subsection{Piece-wise constant forcing}\label{sec:piecewise} 

For most forcing functions $p$, the spectral curves and the first Bloch eigenfunction can only be computed numerically. However, for a specific choice of $p$, they can be computed analytically. Examples of such a $p$ are the piece-wise constant functions. Following the method in \cite{MCB_Breather_2011, breathers_stab}, we can use Floquet-Bloch analysis to compute the band structure $(\ell,\lambda_n(\ell))_{n \in \mathbb{N}}$ and $\psi_1$ explicitly for this class of functions. 
Such an explicit calculation will be part of future work. Numerical computations for \eqref{EQ:SH_periodic_coefficients} with $p$ constructed periodically as $p(x+2\pi/k_f) =p(x)$, $x\in\mathbb{R}$, and $p(x) =  \mu + \gamma\chi_{[2\pi/3k_f, ~4\pi/3k_f)}$ for $x\in[0, 2\pi/k_f)$ show that the spectrum splits at the intersections again. Compared to the $\cos(k_fx)$ forcing, the shape of $\lambda_n(\ell)$ changes much slower upon increasing $\gamma$. The spectrum for the non-resonant case $k_f=3.5$ and $\gamma=2.2$ is shown in Figure \ref{fig:SH_spec_ppw_kf3p5}.

\begin{figure}[ht]
    \centering
    \includegraphics[width=0.45\linewidth]{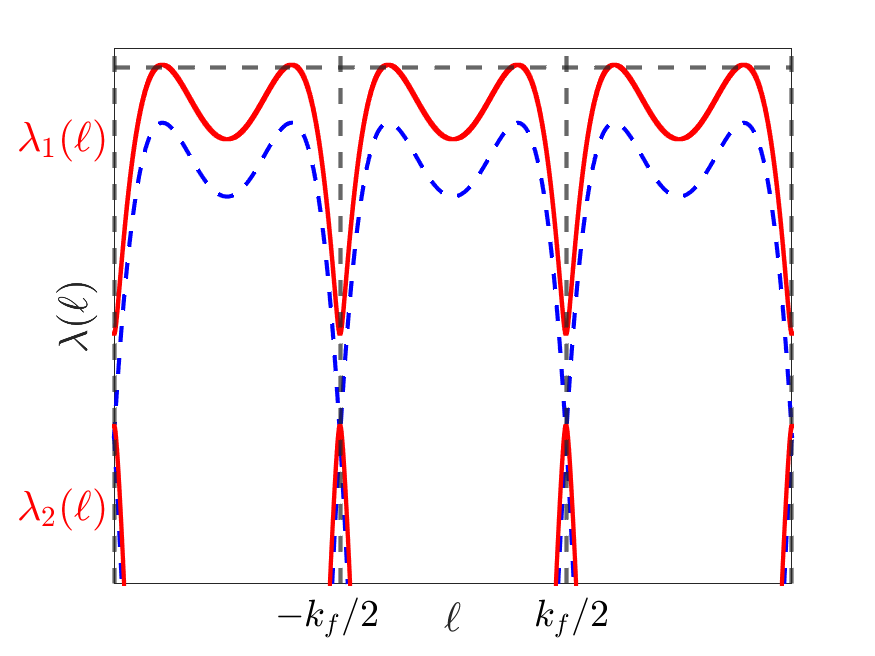}
    \caption{Spectrum of Swift-Hohenberg equation, forced with a piecewise constant forcing function with parameters $k_f=3.5$, $\mu=-0.75$ and $\gamma=2.2$.}
    \label{fig:SH_spec_ppw_kf3p5}
\end{figure}

\subsection{Non-periodic forcing}
Even when forcing with a non-periodic function $p$, numerical studies show that the zero state bifurcates into a pattern solution, where patterns arise at places where $p(x)>0$. This phenomenon is shown in Figure \ref{fig:SH_time_evol_ter1}, for an artificially constructed terrain function $p(x)= \mu + \gamma \widetilde{p}(x)$ with $\widetilde{p}$ non-periodic. 

\begin{figure}[ht]
\centering
    \begin{subfigure}[b]{0.4\textwidth}
    \includegraphics[width=\textwidth]{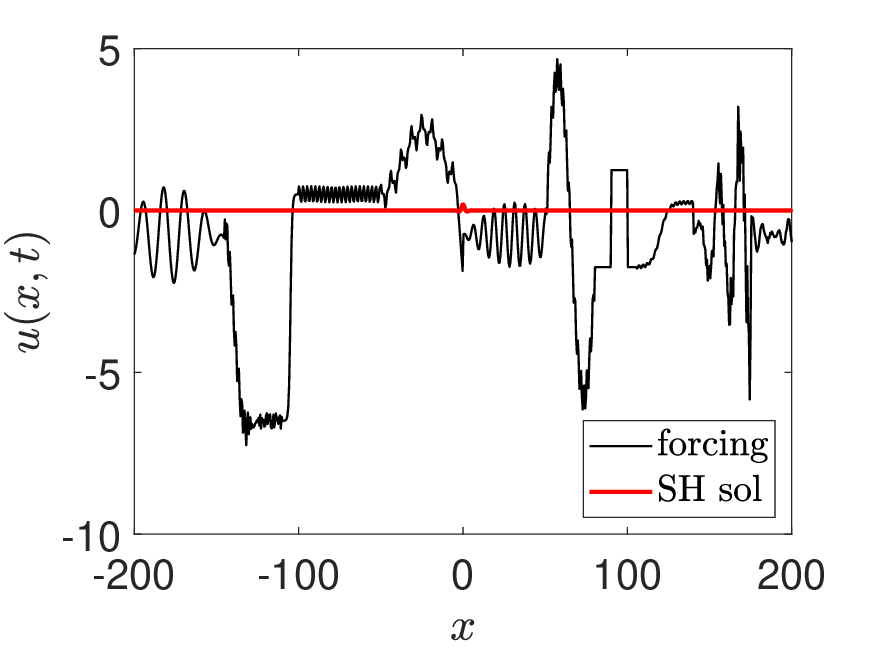}
    \caption{Initial condition ($t=0$)}
  \end{subfigure}
      \begin{subfigure}[b]{0.4\textwidth}
    \includegraphics[width=\textwidth]{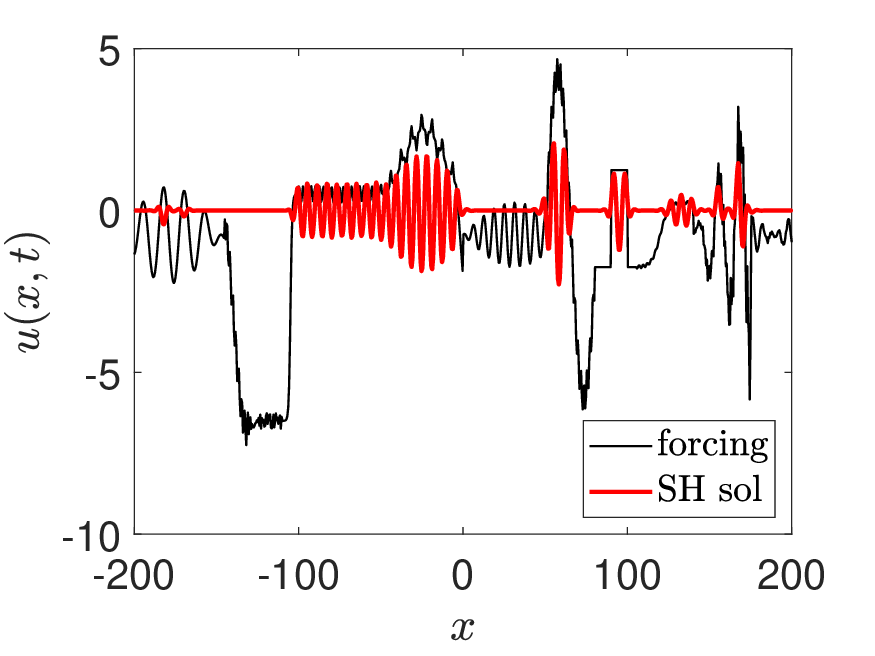}
    \caption{Solution at time $t=500$}
  \end{subfigure}
  \caption{Solution in time of system \eqref{EQ:SH_periodic_coefficients} with $p(x)$ non-periodic in red, forcing $p(x)$ plotted in black.}
    \label{fig:SH_time_evol_ter1}
\end{figure}

\subsection{Multidimensional case}\label{sec:2d}
Pattern formation in higher dimensions is an active area of research with many open questions. Directly related to our research is \cite{Ehud_spatialperforcing}, which took small forcing (specifically $\gamma = \mathcal{O}(\varepsilon^2)$) in a two-dimensional periodically forced Swift-Hohenberg model. The large forcing case has not yet been investigated, and would be a natural next step. We refrain from giving an extensive list of articles investigating pattern formation in the two-dimensional case and simply mention the classic sources \cite{Hoyle2006Pattern} and \cite{book_Ehud} and recent advances such as \cite{Gaff2023} and \cite{TzouTzou2020CurvedTerrain}, \cite{TzouTzou2019CurvedTorus} to which we might find immediate connections.

\subsection{Pattern formation in reaction-diffusion systems}

We have studied the Swift-Hohenberg equation as a paradigm model for pattern formation. Since the applicability of the obtained results to real-world systems is an important consideration, ongoing research addresses this by applying the techniques in this paper to vegetation models with a spatially periodic terrain. To this end, we use the one-dimensional extended Klausmeier model \cite{bastiaansen_doelman}, given by 
\begin{equation}\label{eq:ext_Klausmeier}
\begin{split}
        \partial_t U(x,t) &= D\partial_x^2 U(x,t) + h'(x)\partial_xU(x,t) + h''(x)U(x,t)+a-U(x,t)-U(x,t)V(x,t)^2,\\
        \partial_t V(x,t) &= \partial_x^2 V(x,t) - mV(x,t) +U(x,t)V(x,t)^2, 
\end{split}
\end{equation}
with $x\in\mathbb{R}$, $t\geq 0$, $U=U(x,t), ~V=V(x,t)\in\mathbb{R}$, parameters $D, a, m\in\mathbb{R}_{>0}$ and $h\in C_b^3(\mathbb{R})$. Here, $U$ represents water and $V$ represents vegetation. The amount of rainfall is represented by parameter $a$ and the mortality rate of the vegetation by $m$. The topography is represented by $h(x)$, yielding the spatially varying coefficients. In the constant coefficient case, the system has three background states: one stable bare soil state and two vegetated states $(U_\pm, V_\pm)$, with $(U_-, V_-)$ always being unstable. The spectral picture of the linear operator for the constant coefficient case around $(U_+, V_+)$ bears similarities with the spectral picture of the constant coefficient Swift-Hohenberg system and solutions after the Turing bifurcation can be approximated by solutions of a Ginzburg-Landau equation
\cite{GKGS_SDHR}. Although the spectrum of the constant coefficient case closely resembles that of our paradigm model, it is more complex. In the constant coefficient case, the location of the maximum $k_0$ shifts as the system passes through onset. After introducing $h$, the background states change. In particular, the new background state is no longer constant. For $h(x) = \gamma \cos(k_f x)$, with forcing strength $\gamma$ and different wavenumbers $k_f$, the new background state is shown in Figure \ref{fig:KM_BS}. Additionally, the influence of the term $\partial_x U$ is new compared to our paradigm model, resulting in an imaginary component in the spectrum. Addressing these changes will be the focus of future research.

\begin{figure}[ht]
\centering
    \begin{subfigure}[]{0.4\textwidth}
    \includegraphics[width=\textwidth]{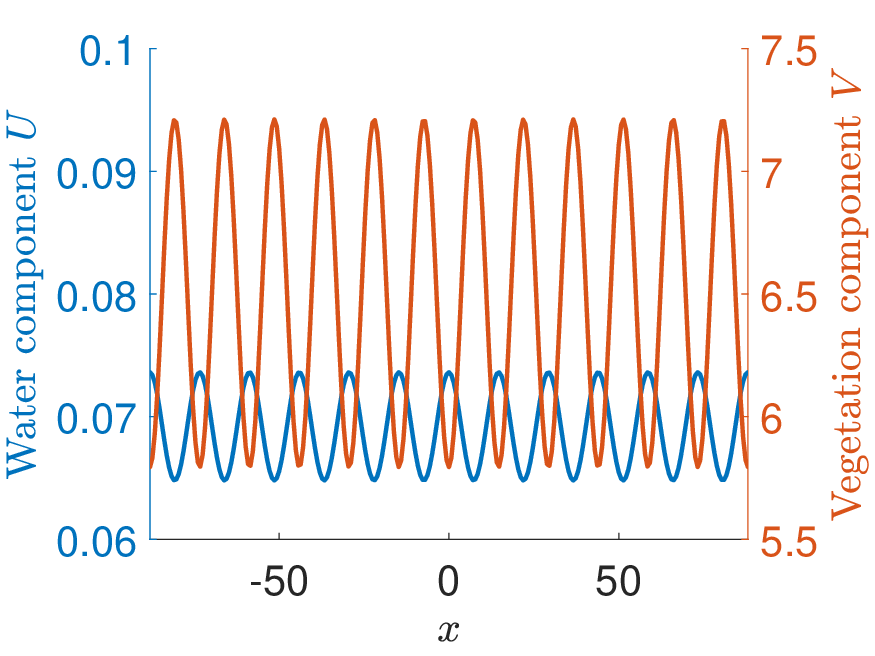}
    \caption{Background state for $k_f=k_0$}
  \end{subfigure}
    \begin{subfigure}[]{0.4\textwidth}
    \includegraphics[width=\textwidth]{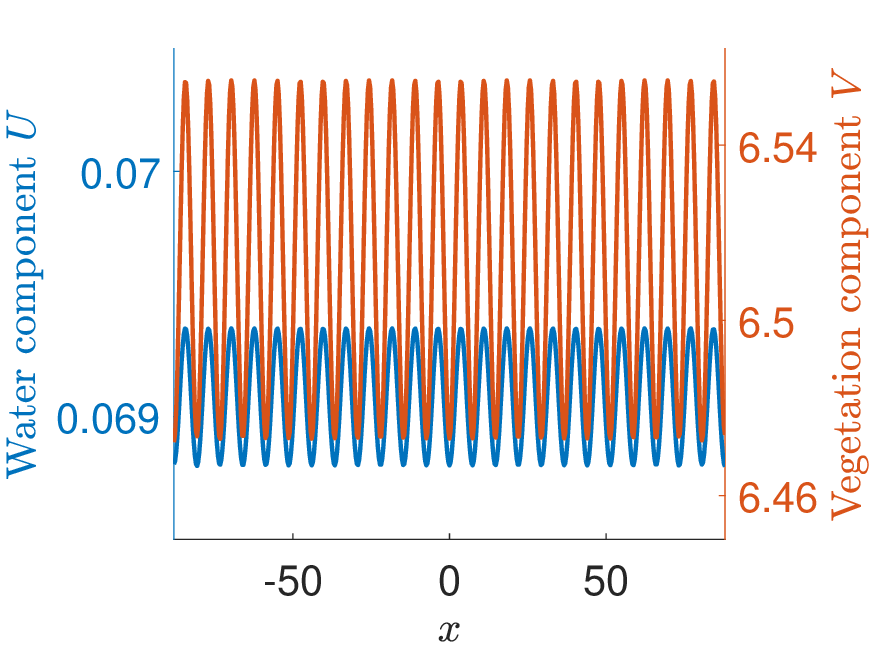}
    \caption{Background state for $k_f=2k_0$}
    \end{subfigure}
  \caption{Background state of extended Klausmeier model with $h(x) = \gamma \cos(k_fx)$, for $a=3$, $m=0.45$, $D=500$, $\gamma = 3$ and different values of $k_f$ relative to $k_0$.}
  \label{fig:KM_BS}
\end{figure}

\textbf{Acknowledgement.} The authors acknowledge the support of the research through NWO grant 613.009.154.

\newpage
\setlength\bibitemsep{0.4\baselineskip}
\printbibliography[heading=bibintoc,title={Bibliography}]

\end{document}